\documentclass[12pt,english]{article}
\usepackage{lmodern}
\usepackage[T1]{fontenc}
\usepackage[latin9]{inputenc}
\usepackage{geometry}
\usepackage{color}
\usepackage{babel}
\usepackage{amsmath}
\usepackage{amsthm}
\usepackage{amssymb}
\usepackage{float,hhline,microtype,afterpage,soul,enumitem,pgfplots,pdfpages,verbatim,multirow,booktabs,subcaption}
\usepackage{graphicx}
\usepackage{setspace}
\usepackage[authoryear]{natbib}
\usepackage[unicode=true,
 bookmarks=true,bookmarksnumbered=false,bookmarksopen=false,
 breaklinks=false,pdfborder={0 0 1},backref=false,colorlinks=true]
 {hyperref}
\hypersetup{pdftitle={Misspecified Learning and Evolutionary Stability},
 pdfborderstyle=,pdfpagelayout=OneColumn,pdfnewwindow=true,pdfstartview=XYZ,plainpages=false,urlcolor=[rgb]{0.0430,0,0.5},linkcolor=[rgb]{0.0430,0,0.5},citecolor=[rgb]{0.0430,0,0.5},hypertexnames=false}

\makeatletter

\providecommand{\tabularnewline}{\\}

\theoremstyle{definition}
\newtheorem{defn}{\protect\definitionname}
\theoremstyle{definition}
 \newtheorem{example}{\protect\examplename}
\theoremstyle{plain}
\newtheorem{thm}{\protect\theoremname}
\theoremstyle{plain}
\newtheorem{prop}{\protect\propositionname}
\theoremstyle{definition}
\newtheorem{condition}{\protect\conditionname}
\theoremstyle{remark}

\theoremstyle{plain}
\newtheorem{lem}{\protect\lemmaname}
\theoremstyle{plain}
\newtheorem{assumption}{\protect\assumptionname}
\theoremstyle{remark}

\usepackage{color}
\usepackage{babel}

\providecommand{\assumptionname}{Assumption}
\providecommand{\claimname}{Claim}
\providecommand{\conditionname}{Condition}

\providecommand{\definitionname}{Definition}
\providecommand{\examplename}{Example}
\providecommand{\lemmaname}{Lemma}
\providecommand{\notationname}{Notation}
\providecommand{\propositionname}{Proposition}
\providecommand{\theoremname}{Theorem}

\makeatother

\providecommand{\assumptionname}{Assumption}
\providecommand{\claimname}{Claim}
\providecommand{\conditionname}{Condition}
\providecommand{\definitionname}{Definition}
\providecommand{\examplename}{Example}
\providecommand{\lemmaname}{Lemma}
\providecommand{\notationname}{Notation}
\providecommand{\propositionname}{Proposition}
\providecommand{\theoremname}{Theorem}

\begin{document}
\title{\vspace{-60bp}
Misspecified Learning and Evolutionary Stability\thanks{We thank Cuimin Ba, Thomas Chaney, Sylvain Chassang, In-Koo Cho, Krishna Dasaratha, Andrew Ellis, Ignacio
Esponda, Mira Frick, Drew Fudenberg, Alice Gindin, Ryota Iijima, Yuhta
Ishii, Philippe Jehiel, Pablo Kurlat, Elliot Lipnowski, Jonny Newton, Filippo Massari,
Andy Postlewaite, Philipp Sadowski, Alvaro Sandroni, Grant Schoenebeck,
Joshua Schwartzstein, Philipp Strack, Carl Veller, and various conference and seminar
participants for helpful comments.  Byunghoon Kim provided excellent research assistance. Kevin He thanks the California
Institute of Technology for hospitality when some of the work on this
paper was completed, and also acknowledges the University Research Foundation Grant at the University of Pennsylvania for financial support. Jonathan Libgober thanks Yale University and the Cowles Foundation for their hospitality.}}
\author{Kevin He\thanks{University of Pennsylvania. Email: \texttt{\protect\protect\protect\href{mailto:hesichao\%5C\%5C\%5C\%40gmail.com}{hesichao@gmail.com}}}
\and Jonathan Libgober\thanks{University of Southern California. Email: \texttt{\protect\protect\protect\href{mailto:libgober\%5C\%5C\%5C\%40usc.edu}{libgober@usc.edu}}}}
\date{{\normalsize{}{}{}}%
\begin{tabular}{rl}
First version: & December 20, 2020\tabularnewline
This version: & \today\tabularnewline
\end{tabular}}

\maketitle
\vspace*{-20bp}

\begin{abstract}
{\normalsize{}{}{}\thispagestyle{empty} \setcounter{page}{0}}{\normalsize\par}

\noindent We extend the indirect evolutionary approach to the selection of (possibly
misspecified) models. Agents with different models match in pairs to play
a stage game, where models define feasible beliefs about game parameters
and about others' strategies. In equilibrium, each agent adopts the feasible
belief that best fits their data and plays optimally given
their beliefs. We define the stability of the resident model by comparing
its equilibrium payoff with that of the entrant model, and provide
conditions under which the correctly specified resident model can
only be destabilized by misspecified entrant models that contain multiple
feasible beliefs (that is, entrant models that permit inference).
We also show that entrants may do well in their matches against the
residents only when the entrant population is large, due to the endogeneity
of misspecified beliefs. Applications include the selection of demand-elasticity
misperception in Cournot duopoly and the emergence of analogy-based
reasoning in centipede games.   \medskip{}

\newpage{}
\end{abstract}

\newpage

\section{Introduction}

In many economic settings, people draw \emph{misspecified inferences}
about the world: they learn from data but exclude the true data-generating process from consideration. \citet{esponda2016berk} introduce Berk-Nash equilibrium to accommodate this observation --- a  solution
concept where agents use data to infer the best-fitting mapping from
actions to outcomes, out of a set of  mappings that are all wrong.
A line of research (discussed below) uses this and related solution
concepts to study the implications of Bayesian learning under particular
misspecifications, with most papers treating misspecifications as
exogenously given.

When should we expect misspecified inference to take hold in a rational
society, as assumed in much of this literature? A defining feature that distinguishes misspecified inference from other kinds of errors and biases is the use of data to form beliefs  --- how does this \emph{belief endogeneity} affect its viability? 
We develop a framework to answer these questions from an evolutionary perspective, studying the objective equilibrium payoffs of different agents. Unlike contemporaneous work in single-agent settings \citep*{FL_mutation,FII_welfare_based}, we focus on strategic interactions where payoffs depend on equilibrium play.

\subsection{Summary of the Setup}

Each agent is endowed with a \emph{model} that contains free parameters; the parameter values  correspond to feasible beliefs about the stage
game and about others' strategies. The model's adherents think that, for some parameter value, 
the instantiated model describes the true
stage game and  opponents' behavior. They estimate the best-fitting parameter
value, which determines their subjective preference. Models rise and
fall in prominence based on  objective equilibrium payoffs of their adherents,
as higher payoffs confer greater evolutionary success.

Society consists of the adherents of multiple competing models, who
match up to play the stage game every period. Agents can identify
which subpopulation their opponent belongs to, and (correctly) know that
the game they play does not depend on opponent type.\footnote{If the players thought that the stage game could change  with the 
opponent, then this would give additional channels for biases to invade
a rational society. Our framework focuses on how the belief endogeneity
that plays a distinctive role in misspecified learning affects the
viability of errors.} Our framework assumes that agents may face one of several possible
stage games, so richer models can in principle help agents by allowing
them to adapt their behavior to the true game. Conditional on a realized stage game,  each agent forms in equilibrium a Bayesian belief
about the game and about others' strategies using data from all their interactions, and plays  a subjective best response to each opponent
subpopulation given this belief.

We say model A is  \emph{evolutionarily stable}  against model
B if, for all sufficiently small population shares of model B, model A's 
equilibrium payoff (averaged across the distribution of stage games) is weakly higher than that of model B. This criterion is familiar from the literature on the \emph{indirect
evolutionary approach}, which considers evolution acting on some trait
that determines agents' strategy choices, as opposed to acting on
these choices directly. Our stability concepts reduce to standard
notions under this approach when models stipulate a dogmatic belief about the stage game. Our main contribution
is to study flexible models that contain multiple feasible beliefs
and, hence,  the role of belief endogeneity in stability. 

\subsection{Belief Endogeneity and the Viability of Misspecifications}

Consider an agent with a flexible model that contains multiple feasible
beliefs. Equilibrium belief is endogenous because the agent infers model parameters from data, and this data depends on population composition and opponents' strategies. By contrast, equilibrium belief is exogenous
for an agent whose model  contains only one feasible belief --- by construction,
data plays no role in shaping this (possibly distorted) belief. Our
formal results identify two novel stability phenomena that can only
arise with belief endogeneity: 
\begin{enumerate}
\item Endogenous beliefs allow agents to adopt different subjective best
responses in different games, so misspecified inference may confer
greater strategic benefits than dogmatic beliefs. 
\item Agents with a fixed misspecification may be weak when rare but strong when common. 
\end{enumerate}
Section \ref{subsec:When-Is-the} discusses  the former point by characterizing
environments where the correctly specified model is only evolutionarily
fragile against invading models that contain multiple feasible beliefs.
One part of our argument constructs an optimal misspecified model
for invading a rational society. This misspecification resembles an
``illusion of control'' bias, where agents think the game's outcome only depends on their own strategy and not on the opponent's
strategy. Adherents of this model end up adopting the optimal commitment against a correctly specified opponent, game by game.

But in some environments, there exists a single commitment strategy
that is beneficial in every game, so belief endogeneity is not necessary
for the invading model. The second part of our argument in Section
\ref{subsec:When-Is-the} identifies a geometric condition that ensures any model with a fixed dogmatic belief across all stage games cannot outperform the correctly specified model. Putting everything together: when the geometric condition holds and the rational model fails to achieve the best commitment payoff in some stage game, misspecified inference is necessary for the rational model's fragility. 

Section \ref{subsec:Stability-Reversals} then turns to the latter point and identifies  a type
of fluidity in a misspecified model's performance based on the population
proportions. Two models are said to exhibit \emph{stability reversal}
if: 
\begin{itemize}
\item Whenever model A is dominant, its adherents strictly outperform model
B's adherents not only on average, but even conditional on the opponent's
type; and
\item Whenever model B is dominant, its adherents strictly outperform model
A's adherents on average.
\end{itemize}
In the absence of belief endogeneity, the first condition would imply
that A outperforms B regardless of the two subpopulations' sizes.
But this no longer holds when belief is endogenous due to misspecified
inference. The reason is that the data from the B-vs-B matches  may induce more evolutionarily
advantageous beliefs than the data from the B-vs-A matches. 

\subsection{Applications}

Extending the indirect evolutionary approach to accommodate models
with belief endogeneity lets us analyze new applications. Our companion paper, \citet{LQNPaper}, illustrates this point by studying the selection of misspecified higher-order beliefs in an incomplete-information Cournot duopoly game. In the present paper, Section \ref{sec:example} presents a simpler complete information duopoly game, while Section \ref{sec:ABEE} considers
the selection of \emph{analogy classes} in extensive-form games \citep{jehiel2005analogy},
based on the payoffs for players with different analogy classes. Under
analogy-based reasoning, players (incorrectly) believe opponents choose the same action distribution at all nodes within an analogy class and infer that distribution from  empirical frequencies. Our approach predicts not only
that analogy-based reasoning may invade a correctly-specified society,
but also that the two can coexist. By solving for the corresponding stable population composition, we obtain sharp predictions on the relative
prominence of analogy-based reasoning as a function of the underlying
stage game.

\subsection{\label{subsec:literature}Related Literature}

Our paper contributes to the literature on misspecified Bayesian learning
by proposing a framework to assess which specifications are more likely
to persist based on their objective payoffs. Misspecified inference in such a framework leads to endogenous beliefs, which in turn generate new  phenomena in  payoff-based selection of biases.  Most prior works on
misspecified Bayesian learning, by contrast, take the misspecification as exogenous, studying the subsequent implications in both single-agent decision problems\footnote{See \citet*{nyarko1991learning,fudenberg2017,heidhues2018unrealistic,he_gambler,cho2025learning}. Also related is \cite{Fudenberg2024} who show how memory limitations yield inferences resembling those in misspecified learning models.}
and multi-agent games.\footnote{See \citet*{bohren2016informational,bohren2017bounded,jehiel2018investment,molavi2019macro,dasaratha2020network,ba2020overconfidence,frick2019misinterpreting,Murooka2023}.}
Several papers establish general convergence properties of misspecified
learning.\footnote{See \citet*{esponda2016berk,esponda2019asymptotic,frick2019stability,FLS_general_conv}.}

Our interest in endogenizing misspecifications using objective payoffs\footnote{A separate line of work that has used objective payoffs to endogenize misspecified inference, restricting attention to financial markets \citep{sandroni2000markets,massari2020under}, while our approach applies to general strategic environments.} contrasts with alternatives using subjective expectations of payoffs\footnote{See \citet*{olea2019competing,Levyetal2020,gagnon2018channeled}}
or goodness-of-fit tests.\footnote{See \citet{cho2015learning,cho2017gresham,ba2020,schwartzstein2020using,Lanzani2022}.}  There are several other papers that also use objective payoffs as the selection criterion. Our work differs in that we assume  agents update beliefs using Bayes' rule and focus on the selection between various misspecifications under Bayesian learning. \citet*{massari2020learning} show that learning rules departing from Bayes' rule can achieve higher objective payoffs, while \citet*{massari2024rational} show that maximizing a combination of accuracy and payoffs can improve performance along \emph{both} dimensions. Similarly, \citet*{heller2020biased} and \citet*{berman2020naive} study the evolution of different belief-formation processes under a reduced-form (and possibly non-Bayesian) approach,
considering arbitrary inference rules.

Two independent and contemporaneous papers,
\citet*{FL_mutation} and \citet*{FII_welfare_based}, also consider
payoff-based criteria under Bayesian inferences for selecting misspecifications, but restrict attention to single-agent
decision problems.\footnote{\citet*{FL_mutation} study a framework where a
continuum of agents with heterogeneous misspecifications arrive each
period and learn from their predecessors' data. \citet*{FII_welfare_based}
assign a \emph{learning efficiency index} to every misspecified signal
structure and conduct a robust comparison of welfare under different
misspecifications.} We differ in highlighting that belief endogeneity can \emph{strictly} expand the possibility for misspecifications to invade rational societies in strategic settings (relative to biased invaders who do not draw inferences).

Our framework of competition between different specifications for
Bayesian learning is inspired by the evolutionary game theory literature. Relative to this literature, our contribution is to accommodate misspecified inference. 
We follow past work that also
uses objective payoffs as the selection criterion for subjective preferences in games and decision problems
(e.g., \citet*{Dekeletal2007}, see also the surveys \citet*{robson2011evolutionary}
and \citet*{Alger2019survey}) and the evolution of constrained strategy
spaces \citep*{heller2015three,heller2016rule}. Like us, \citet*{GuthNapel2006} allow for stage-game heterogeneity, studying the ability to discriminate between these games.

\section{Environment and Stability Concept \label{sec:Environment-and-Stability}}

We start with our formal stability concept, defining \emph{equilibrium
zeitgeist} to determine the evolutionary fitness of specifications
that coexist in a society. Our general setup allows agents to both
learn about the fundamentals and draw inferences about others' strategies:
indeed, misinference about opponent's strategy is central to our application
in Section \ref{sec:ABEE}. But most of our results and applications
concern a special case of the setup where agents correctly know others'
strategies in equilibrium, so the focus is on misinference about fundamental
uncertainty and the role of such misinference on evolutionary selection.
In the main text we primarily focus on the steady-state characterization
of equilibrium zeitgeists, but we provide a learning foundation for
this solution concept in Appendix \ref{sec:Learning-Foundation}. 

\subsection{Objective Primitives}\label{subsec:Objective-Primitives}

Agents in a population repeatedly match to play a stage game, which
is a symmetric two-player game with a common, metrizable strategy
space $\mathbb{A}$. There is a set of possible states of nature $G\in\mathcal{G}$,
called \emph{situations}. The strategy choices $a_{i},a_{-i}\in\mathbb{A}$
of $i$ and $-i$, together with the situation, stochastically generate
consequences $y_{i},y_{-i}\in\mathbb{Y}$ from a metrizable space
$\mathbb{Y}$. Each agent $i$'s consequence $y_{i}$ determines their
utility, according to a common utility function $\pi:\mathbb{Y}\to\mathbb{R}$, which
we take to be Borel measurable with respect to the sigma algebra generated
by the topology on $\mathbb{Y}$. The objective distribution over
consequences is $F^{\bullet}(a_{i},a_{-i},G)\in\Delta(\mathbb{Y}),$
with an associated density or probability mass function denoted by
$f^{\bullet}(a_{i},a_{-i},G),$ where $f^{\bullet}(a_{i},a_{-i},G)(y)\in\mathbb{R}_{+}$
for each $y\in\mathbb{Y}$. We suppress $G$ from $f^{\bullet}$ and
$F^{\bullet}$ when $|\mathcal{G}|=1$. We allow for $\mathbb{Y}$
to be general outside of the previous technical restrictions.

This setup captures mixed strategies (if $\mathbb{A}$ is the set
of mixtures over some pure actions), incomplete-information games
(if $S$ is a space of private signals, $A$ a space of actions, and
$\mathbb{A}=A^{S}$ is the set of signal-contingent actions), and
even asymmetric games. For the latter, we consider the ``symmetrized''
version where each player is placed into each role with equal probability
(see Section \ref{sec:ABEE} for one application where agents play
an asymmetric game).

In addition, at the end of a match where the strategy profile $(a_{i},a_{-i})$
is played, each agent $i$ observes a \emph{monitoring signal} $m_{i}$
about the opponent's strategy which only depends on $a_{-i}$ and not
on $a_{i}$ or the situation. Let $\mathbb{M}$ be the space of monitoring
signals, and let the objective distribution over monitoring signals when the opponent plays $a_{-i}$
be given by the density or probability mass function $\varphi^{\bullet}(a_{-i})$,
where $\varphi^{\bullet}(a_{-i}):\mathbb{M} \to \mathbb{R}_{+}$. The monitoring
signal $m_{i}$ is payoff irrelevant and is generated independently
of the consequences. Our framework separately defines the monitoring signal for the expositional simplicity of introducing the special case of environments with strategic certainty (where the monitoring signal is perfect) and for discussing the learning foundation of equilibrium in such environments (where we make the monitoring signal ``almost perfect'').

\subsection{Models and Parameters}

Throughout this paper, we will take the strategy space $\mathbb{A},$
the set of consequences $\mathbb{Y},$ the utility function over consequences
$\pi$, the set of monitoring signals $\mathbb{M}$ and the strategy
monitoring structure $\varphi^{\bullet}$ to be
common knowledge among the agents. But, agents are unsure about how
play in the stage game translates into consequences: that is, they
have \emph{fundamental uncertainty} about the function $(a_{i},a_{-i})\mapsto F^{\bullet}(a_{i},a_{-i},G).$ While we assume that the situation $G$ is unobserved, we allow agents
to draw inferences about it by observing the consequences from the
matches they face. Agents may also be unsure about which strategies
others use (\emph{strategic uncertainty}), but they could get some information about others' play  through their consequences and monitoring signals. 

We focus on the case where society consists of two\footnote{We view the case of two groups of agents with different models as
the natural starting point, though it is straightforward to generalize
Definition \ref{def:Zeitdef} to the case of more than two groups.} observably distinguishable groups of agents, A and B, who may behave
differently in the stage game due to different beliefs about how $y$
is generated and about the strategies of their opponents. The two
groups of agents entertain different \emph{models} of the world that
help resolve their fundamental uncertainty and strategic uncertainty.
A model $\Theta$ is a collection of \emph{parameters} $(a_{A},a_{B},F)$
with $a_{A},a_{B}\in\mathbb{A}$ and $F:\mathbb{A}^{2}\to\Delta(\mathbb{Y})$. 
So, each parameter specifies conjectures $a_{A},a_{B}$ about how
group A and group B opponents will act when playing against the agent. It also contains a conjecture
$F$ about how strategy profiles translate into consequences for the
agent. So, we can view each model as a subset of $\mathbb{A}^{2}\times(\Delta(\mathbb{Y}))^{\mathbb{A}^{2}}$.
We assume the marginal of the model on $(\Delta(\mathbb{Y}))^{\mathbb{A}^{2}}$
is indexed by some $\gamma\in\Gamma$ for a metric space $\Gamma$
acting as an indexing set, so this marginal can be written as $\{F_{\gamma};\gamma\in\Gamma\}.$
For each $F_{\gamma}$ that is part of some parameter and for every
$(a_{i},a_{-i})\in\mathbb{A}^{2}$, we suppose $F_{\gamma}(a_{i},a_{-i})$
is a Borel measure on $\mathbb{Y}$ and it has associated with it
a density or probability mass function $f_{\gamma}(a_{i},a_{-i}):\mathbb{Y}\to\mathbb{R}_{+}$.
We also suppose that for every $(a_{i},a_{-i}),$ the map $\gamma\mapsto\mathbb{E}_{y\sim F_{\gamma}(a_{i},a_{-i})}[\pi(y)]$
is Borel measurable.\footnote{Note that this measurability property would follow from the measurability
of the mapping $\gamma\mapsto f_{\gamma}(a_{i},a_{-i})(y)$ for each
fixed $(a_{i},a_{-i},y)$, under some further restrictions necessary
to apply Fubini's theorem to the function $(a_{i},a_{-i},y)\mapsto\pi(y)f(a_{i},a_{-i})(y)$.}

Each agent enters society with a persistent model, which depends entirely
on whether she is from group A or group B. We refer to the agents
who are endowed with a given model  the \emph{adherents} of that
model. We call a model \emph{correctly specified }if it is a superset
of $\mathbb{A}^{2}\times\{F^{\bullet}(\cdot,\cdot,G):G\in\mathcal{G}\}$
, so the agent can make unrestricted inferences about others' strategies
and does not rule out the correct data-generating process $F^{\bullet}(\cdot,\cdot,G)$
for any situation $G$. We call $\Theta=\mathbb{A}^{2}\times\{F^{\bullet}(\cdot,\cdot,G):G\in\mathcal{G}\}$
the \emph{minimal correctly specified }model. A model may exclude
the true $F^{\bullet}(\cdot,\cdot,G)$ that produces consequences,
at least in some situation $G$, or it may exclude some strategies
as feasible conjectures of others' play. In this case, the model is
\emph{misspecified}.

An important special case of the setup focuses purely on misinference
about fundamental uncertainty. 
\begin{defn}
An environment has \emph{strategic certainty} if
\end{defn}

\begin{itemize}
\item $\mathbb{M}=\mathbb{A}$ and $\varphi^{\bullet}(a_{-i})$ puts probability
1 on $a_{-i}$ for every $a_{-i}\in\mathbb{A}$,
\item The model of each group $g\in\{A,B\}$ is of the form $\mathbb{A}^{2}\times\mathcal{F}_{g}$
for some $\mathcal{F}_{g}\subseteq(\Delta(\mathbb{Y}))^{\mathbb{A}^{2}}$,
and 
\item Every $y\in\mathbb{Y}$ with the property that $f^{\bullet}(a_{i}',a_{-i}',G)(y)>0$
for some $a_{i}',a_{-i}'\in\mathbb{A}$ and some $G\in\mathcal{G}$
also satisfies $f(a_{i}',a_{-i}'')(y)>0$ for every $a_{-i}''\in\mathbb{A}$
and every $f$ that is the density or probability mass function of
some conjecture $F\in\mathcal{F}_{A}\cup\mathcal{F}_{B}$. 
\end{itemize}

In environments with strategic certainty, monitoring signals perfectly
reveal opponent's strategy and all agents can make unrestricted strategic
inferences. Combined with the assumption that all consequences that
agents can observe when they play $a_i'$ have positive likelihood under any of their feasible
conjectures about fundamental uncertainty, this will imply that agents
hold correct beliefs about their opponents' strategies in the equilibrium
concept that we define below. In such environments, the focus is on
different groups' feasible conjectures about the fundamental uncertainty,
$\mathcal{F}_{A}$ and $\mathcal{F}_{B}$. We will therefore sometimes
omit mention of the monitoring signal when analyzing environments
with strategic certainty. 

In environments with strategic certainty, a model $\Theta=\mathbb{A}^{2}\times\mathcal{F}$
with $|\mathcal{F}|=1$ is called a \emph{singleton} model. This terminology
refers to the fact that such models stipulate a single dogmatic belief
about the fundamental uncertainty (though agents still make flexible
inferences about others' strategies). Adherents of singleton models
do not draw inferences about the game from data.  Since the situation is itself
unobserved, this implies such agents also do not change their preferences with the situation, since this is only possible through
drawing different inferences in different situations.

\subsection{Zeitgeists}\label{subsec:ZeitDef}

To study competition between two models, we must describe the social
composition and interaction structure in the society where learning
takes place. We have in mind a setting where each agent plays the
stage game with a uniformly random opponent in every period and uses
their personal experience in these matches to calibrate the most accurate
parameter within their model. A zeitgeist describes the corresponding
landscape. 
\begin{defn}
\label{def:Zeitdef}Fix models $\Theta_{A}$ and $\Theta_{B}$. A
\emph{zeitgeist $\mathfrak{Z}=(\mu_{A}(G),\mu_{B}(G),p,a(G))_{G\in\mathcal{G}}$
}consists of: (1) for each situation $G,$ a belief over parameters
for each model, $\mu_{A}(G)\in\Delta(\Theta_{A})$ and $\mu_{B}(G)\in\Delta(\Theta_{B})$;
(2) relative sizes of the two groups in the society, $p=(p_{A},p_{B})$
with $p_{A},p_{B}\ge0,$ $p_{A}+p_{B}=1$; (3) for each situation
$G,$ each group's strategy when matched against each other group,
$a=(a_{AA}(G),a_{AB}(G),a_{BA}(G),a_{BB}(G))$ where $a_{g,g^{'}}(G)\in\mathbb{A}$
is the strategy that an adherent of $\Theta_{g}$ plays against an
adherent of $\Theta_{g^{'}}$ in situation $G$. 
\end{defn}

A zeitgeist outlines the beliefs and interactions among agents with
heterogeneous models living in the same society. Part (1) captures
the belief of each group. Part (2) determines the relative prominence
of each model. Agents are matched uniformly at random across the entire
society, so an agent from group $g$ has probability $p_{g}$ of being
matched with an opponent from their own group and a complementary chance
of being matched with an opponent from the other group.\footnote{\cite{LQNPaper} contains an example that varies the matching assortativity
and discusses how this affects the selection of biases.} Part (3) describes behavior in the society. Note that a zeitgeist
describes each group's situation-contingent belief and behavior, since
agents may infer different parameters and thus adopt different subjective
best replies in different situations. However, we emphasize
that since  situations are not directly observed, they only influence
strategies by changing the distribution of the agents'  consequences (and hence
their beliefs).

\subsection{Equilibrium Zeitgeists}\label{subsec:EZDef}

A model's fitness corresponds to the equilibrium payoff of its adherents.
An equilibrium zeitgeist (EZ) requires behavior to be optimal given beliefs and beliefs to 
best fit the data given behavior. As we  make clear in the learning foundation for
EZs in Appendix \ref{sec:Learning-Foundation}, this equilibrium concept relates to steady states
in a  society of long-lived Bayesian learners who use consequences and monitoring signals
to make Bayesian inferences among parameters in their model, assuming there is convergence in
beliefs and  behavior.  In the steady state, agents choose subjectively optimal strategies
given their beliefs about others' strategies and about the stage game.

We now formalize this criterion. For two distributions over consequences
and monitoring signals, $\Phi,\Psi\in\Delta(\mathbb{Y}\times\mathbb{M})$
with density or probability mass functions $\phi,\psi$, define the Kullback-Leibler divergence
(KL divergence) from $\Psi$ to $\Phi$ as $D_{KL}(\Phi\parallel\Psi):=\int\phi(y,m)\ln\left(\frac{\phi(y,m)}{\psi(y,m)}\right)d(y,m)$.
Recall that every data-generating process $F$, like the true fundamental
$F^{\bullet}(\cdot,\cdot,G)$, outputs a distribution over consequences
for every strategy profile, $(a_{i},a_{-i})\in\mathbb{A}^{2}$. 
\begin{defn}
\label{def:EZ}A zeitgeist $\mathfrak{Z}=(\mu_{A}(G),\mu_{B}(G),p,a(G))_{G\in\mathcal{G}}$
is an\emph{ equilibrium zeitgeist (EZ)} if, for every $G\in\mathcal{G}$
and $g,g^{'}\in\{A,B\},$ $a_{g,g^{'}}(G)\in\underset{\hat{a}\in\mathbb{A}}{\arg\max}\ \mathbb{E}_{(a_{A},a_{B},F)\sim\mu_{g}(G)}\left[\mathbb{E}_{y\sim F(\hat{a},a_{g^{'}})}(\pi(y))\right]$
and, for every $g\in\{A,B\},$ the belief $\mu_{g}(G)$ is supported
on 
\begin{align*}
\hspace{-12mm} \underset{(\hat{a}_{A},\hat{a}_{B},\hat{F})\in\Theta_{g}}{\arg\min}\left\{ \begin{array}{c}
(p_{g})\cdot D_{KL}(F^{\bullet}(a_{g,g}(G),a_{g,g}(G),G)\times\varphi^{\bullet}(a_{g,g}(G))\parallel\hat{F}(a_{g,g}(G),\hat{a}_{g})\times\varphi^{\bullet}(\hat{a}_{g}))\\
+(1-p_{g})\cdot D_{KL}(F^{\bullet}(a_{g,-g}(G),a_{-g,g}(G),G)\times\varphi^{\bullet}(a_{-g,g}(G))\parallel\hat{F}(a_{g,-g}(G),\hat{a}_{-g})\times\varphi^{\bullet}(\hat{a}_{-g}))
\end{array}\right\} 
\end{align*}
where $-g$ means the group other than $g$ and $\times$ indicates
the product between a distribution on $\mathbb{Y}$ and a distribution
on $\mathbb{M}$. When $p_g=0$ or $(1-p_g)=0$ but it is multiplied by infinity, we use the convention that $0\cdot\infty=\infty$. 
\end{defn}

This definition requires agents from each group $g$ to choose a subjective
best response against their opponents, given the belief $\mu_{g}$
about the fundamental uncertainty and strategic uncertainty. No matter
which group the agent is matched against, these choices are always
made to selfishly maximize their (individual) subjective utility function.
Each agent's belief $\mu_{g}$ is supported on the parameters in their
model that minimize a weighted KL-divergence objective in situation
$G$, with the data from each type of match weighted by the probability
of confronting this type of opponent. The use of KL-divergence minimization
as the inference procedure is standard in the misspecified Bayesian
learning literature (such as in \cite{esponda2016berk}) and goes back to the basic result from 
\cite{berk1966limiting} that the Bayesian posteriors under misspecification
concentrate in the long run on the KL-divergence minimizers. 
We assume inference occurs separately across situations. This reflects
situation persistence, with agents having enough data to establish
new beliefs and behavior before the situation  changes. Our learning
foundation in Appendix \ref{sec:Learning-Foundation} justifies this
situation-by-situation updating, but we omit the details here as it
otherwise plays no role in our results.

In general, agents choose their best-fitting model parameters based
on two kinds of data: consequences and monitoring signals. In
environments with strategic certainty, there is no equilibrium zeitgeist
where the equilibrium belief $\mu_{g}(G)$ for any group $g$ in any
situation $G$ puts positive weight on any parameter $(\hat{a}_{A},\hat{a}_{B},\hat{F})$
where $\hat{a}_{g'}\ne a_{g',g}(G)$ for any groups $g'\in\{A,B\}$.  
This is because any such parameter has a weighted KL divergence of
infinity (given that $\varphi^{\bullet}$ is perfectly informative
about the opponent's strategy), whereas parameters with $\hat{a}_{A}=a_{A,g}(G)$
and $\hat{a}_{B}=a_{B,g}(G)$ have finite weighted KL divergence.
So, in environments of strategic certainty, we can view beliefs in
equilibrium zeitgeists $\mu_{g}(G)$ as simply beliefs over the fundamental
uncertainty $\mathcal{F}_{g}$, with $\mu_{g}(G)$ supported on
\begin{align*}
\underset{\hat{F}\in\mathcal{F}_{g}}{\arg\min}\left\{ \begin{array}{c}
(p_{g})\cdot D_{KL}(F^{\bullet}(a_{g,g}(G),a_{g,g}(G),G)\parallel\hat{F}(a_{g,g}(G),a_{g,g}(G)))\\
+(1-p_{g})\cdot D_{KL}(F^{\bullet}(a_{g,-g}(G),a_{-g,g}(G),G)\parallel\hat{F}(a_{g,-g}(G),a_{-g,g}(G)))
\end{array}\right\} .
\end{align*}
In environments with strategic certainty, we will therefore omit reference
to beliefs about others' strategies in describing zeitgeists and simply
view $\mu_{g}(G)$ as an element in $\Delta(\mathcal{F}_{g}).$ 

\subsection{Evolutionary Stability of Models}

Given a distribution $q\in\Delta(\mathcal{G})$ and an EZ, we define
the \emph{fitness} of each model as the expected objective payoff
of its adherents in the EZ when $G$ is drawn according to $q$. We
have in mind an evolutionary story where the relative success of the
two models depends on their relative fitness: for instance, agents
may play a large number of games in different periods possibly facing
different situations over time, and models of those agents with higher
total objective payoffs are more likely to be adopted in the next
generation.\footnote{One subtlety is that fitness maximization may require not maximizing
expected payoffs, but rather some other function of the distribution
of payoffs, if shocks can be correlated \citep{robson1996biological}.
However, our microfoundation in Appendix \ref{sec:Learning-Foundation}
posits that situations are fixed for long stretches of time, with
no correlated shocks across matches, making the expectation an appropriate
measurement of fitness.} Given this notion of fitness, our question of interest is: Can the
adherents of a \emph{resident model} $\Theta_{A}$, starting at a
position of social prominence, always repel an invasion from a small
 mass of agents who adhere to an \emph{entrant model} $\Theta_{B}$?

Evolutionary stability depends on the fitness of models $\Theta_{A},\Theta_{B}$
in EZs with $p_{A}=1-\epsilon,p_{B}=\epsilon$ for small $\epsilon>0$.
\begin{defn}
\label{def:stability} Say $\Theta_{A}$ is \emph{evolutionarily stable
{[}fragile{]} }against $\Theta_{B}$ if there exists some $\bar{\epsilon}>0$
so that for every $0<\epsilon\le\bar{\epsilon}$, there is at least
one EZ with models $\Theta_{A},\Theta_{B}$, $p=(1-\epsilon,\epsilon)$
and in all such EZs, $\Theta_{A}$ has a weakly higher {[}strictly
lower{]} fitness than $\Theta_{B}$. 
\end{defn}

Evolutionary stability is when $\Theta_{A}$ has higher fitness than
$\Theta_{B}$ in all EZs, and evolutionary fragility is when $\Theta_{A}$
has lower fitness in all EZs.\footnote{If the set of EZs is empty, then $\Theta_{A}$ is neither evolutionarily
stable nor evolutionarily fragile against $\Theta_{B}.$} These two cases give sharp predictions about whether a small share
of entrant-model invaders might grow in size, across all equilibrium
selections. We fix these rather stringent definitions of stability
and fragility, and focus on showing in Section \ref{sec:new_stability_phenomena}
how the belief endogeneity in models can generate new stability
/ fragility phenomena. A third possible case, where $\Theta_{A}$
has lower fitness than $\Theta_{B}$ in some but not all EZs, corresponds
to a situation where the entrant model may or may not grow in the society,
depending on the equilibrium selection.

\subsection{Discussion}

We clarify some important aspects of our framework before proceeding further.

\subsubsection{Equilibrium Zeitgeist Existence} \label{sec:nexistencediscussio}

Implicit in our definition of evolutionary stability and fragility is that we do not have to worry about whether EZs exist in the first place. However, this is not a given, as we have not imposed the continuity and integrability conditions necessary to ensure existence and the well-definedness of expected utilities, KL divergences, and best responses. These issues are familiar from past work, so we do not belabor them in the main text. Appendix \ref{sec:Existence-and-Continuity} provides sufficient conditions that guarantee existence, as well as that the set of EZs is upper hemicontinuous in population shares. This latter result is particularly useful for moving between the $\epsilon >0$ case and the limit as $\epsilon \rightarrow 0$ in the definition of stability. Throughout, we focus on the case where expected utilities, KL divergences, and best responses are well defined and EZs  exist, and refer interested readers to the appendix for the technical conditions that guarantee these.

\subsubsection{Comparison with Other Evolutionary Frameworks}

\label{sec:comparison}

We apply the ``indirect evolutionary approach'' (see \cite{robson2011evolutionary})
to settings where agents can draw inferences (especially misspecified
inferences). In environments with strategic certainty with singleton
models and $|\mathcal{G}|=1$, our framework reduces to the setup
studied by the literature on preference evolution \cite{Alger2019survey},
since singleton models are equivalent to subjective preferences. But
in general, models with multiple parameters allow agents to adapt
their beliefs (which determine their subjective preferences) endogenously.
Allowing for multiple situations is the most direct way for inference
to be beneficial. With only a single situation, any steady-state outcome
that emerges for some model can also emerge with a singleton model.
That said, one could also study settings with multiple situations
without inference (see \cite{GuthNapel2006} for an example of such
an exercise).

\subsubsection{Framework Assumptions}

An important assumption is that agents (correctly) believe the economic
fundamentals (represented by $G$) do not vary depending on which
group they are matched against. That is, the mapping $(a_{i},a_{-i})\mapsto\Delta(\mathbb{Y})$
describes the stage game that they are playing, and agents know that
they always play the same stage game even though opponents from different
groups may use different strategies in the game. As a result, the
agent's experiences in games against both groups of opponents jointly
resolve the same fundamental uncertainty about the environment.\footnote{We note that play between two groups $g$ and $g^{'}$ is not a Berk-Nash
equilibrium \cite{esponda2016berk}, since adherents from one group
draw inferences about the game's parameters from the matches against
the other group, which may adopt a different strategy. A Berk-Nash
equilibrium between groups $g$ and $g'$ would require inferences
to \emph{only} be made from data generated in the match between $g$
and $g'$.} If adherents could believe that the fundamentals can change depending
on their opponent, then this would give a trivial way for in-group
preferences to emerge and also trivialize the question of which errors
could invade. For expositional simplicity, we do not consider this
elaboration.

Our framework
assumes that agents can identify which group their matched opponent
belongs to, though we do not assume that agents know the data-generating
processes contained in other models or that they are capable of making
inferences using other models.  (In other words,  models are primitives and cannot be changed even after agents
see their opponents' actions. Players do not ``read
into'' what others do when learning.) Some other works in the  literature on the indirect evolutionary approach (e.g., \cite{Dekeletal2007})
consider a more general setup where agents only observe their opponent's
group membership with some probability in each match, and receive no information
about their opponent with the complementary probability. We expect the main insights to carry
through when the  probability of observation is high enough. At the other end of the spectrum,
if agents never observe whether their matched opponent is from group A or group B, then results can change dramatically. For instance, consider an environment of strategic certainty with only one situation. There is no EZ where the minimal correctly specified 
resident model has strictly lower fitness than an entrant model. This is because both models face the same distribution of
opponents' strategies and must play a single strategy against this distribution. If the correctly specified model has strictly lower
fitness, then its adherents are not playing an objective best response, which contradicts
the fact that they correctly know the game and have correct beliefs about others' strategies in EZ under strategic certainty.

Even as agents update their beliefs and optimize their behavior,
population proportions $p_{A}$ and $p_{B}$ remain fixed. We imagine
a world where the relative prominence of models changes much more
slowly than the rate of convergence to an EZ. This assumption about
the relative rate of change in the population sizes follows the previous
work on evolutionary game theory (See \cite{Sandholm2001Preference}
or \cite{Dekeletal2007}).

\section{Illustrative Example}
\label{sec:example}

In this section, we use an example to illustrate our framework, explain
the different stability implications of entrant models that expand
the resident model versus entrant models that shift the resident model,
and preview an ``illusion of control'' entrant model that will play
a key role in our general results. 

Consider an environment with strategic certainty and only one situation (so we omit mention of $G$). Suppose each player $i$ is a Cournot duopolist with constant marginal
cost $c$. Each player $i$ simultaneously chooses a quantity $a_{i}$.
A random market price $P=\beta^{\bullet}-r^{\bullet}(a_{1}+a_{2})+\varepsilon$
realizes and is observed by both players, where $r^{\bullet}>0$,
$\beta^{\bullet}>c$ are constants and $\varepsilon$ is a mean-zero
random variable with full support on $\mathbb{R}$. The utility of
player $i$ is $a_{i}\cdot(P-c).$ 

Mapping back into the formalism from Section \foreignlanguage{american}{\ref{sec:Environment-and-Stability},
we can let $i$'s consequence be $y_{i}=(a_{i},P)$ so that $i$'s
utility is a function of the consequence with $\pi(y_{i})=\pi(a_{i},P)=a_{i}\cdot(P-c)$.
The true distribution over consequences $F^{\bullet}(a_{i},a_{-i})$
given the strategy profile $(a_{i},a_{-i})$ is such that the first
dimension of $y_{i}$ is always $a_{i}$, while the second dimension
is distributed according to $\beta^{\bullet}-r^{\bullet}(a_{i}+a_{-i})+\varepsilon$. }

For $r>0$ and $\beta\in\mathbb{R},$ let $F_{r,\beta}$ represent
the conjecture about the stage game where for each strategy profile
$(a_{i},a_{-i}),$ $F_{r,\beta}(a_{i},a_{-i})$ is the distribution
over consequences with the first dimension always being $a_{i}$ and
the second dimension being distributed according to $\beta-r(a_{i}+a_{-i})+\varepsilon$.
The residents (group A) are correctly specified, so their model has
$\mathcal{F}_{A}=\{F_{r^{\bullet},\beta}:\beta\in\mathbb{R}\}$. We
say the entrants (group B) have a \emph{slope perception} of $\hat{r}$
if their model has $\mathcal{F}_{B}=\{F_{\hat{r},\beta}:\beta\in\mathbb{R}\}$.
The idea is that all agents hold dogmatic beliefs about the slope
of the demand curve and use market-price data to make inferences about
the intercept of the demand curve. The residents' belief about the
slope is correct, while the entrants' belief is possibly wrong.

When entrants misperceive the slope to be $\hat{r}\ne r^{\bullet}$,
they misunderstand how quantity choices affect market prices and thus
misinfer the intercept of the demand curve. This in turn distorts
their behavior: we can show that an agent who believes the slope and
intercept of the demand curve to be $r$ and $\beta$ has a subjective
best response of $\frac{\beta-c-ra_{-i}}{2r}$ when their opponent
plays $a_{-i}$. Proposition \ref{prop:cournot_example} summarizes
how the entrants' slope misperception affects their equilibrium behavior
and welfare in the equilibrium zeitgeist with $p_{A}=1,$ $p_{B}=0$. 
\begin{prop}
\label{prop:cournot_example}Suppose $p_{A}=1,p_{B}=0,$ and the residents
are correctly specified. 
\begin{enumerate}
\item In every equilibrium zeitgeist, we have $a_{AA}=\frac{\beta^{\bullet}-c}{3r^{\bullet}}$
and the fitness of the residents is $\frac{(\beta^{\bullet}-c)^{2}}{9r^{\bullet}}.$ 
\item In every equilibrium zeitgeist, the fitness of the entrants is a function
of their equilibrium strategy: $\frac{1}{2}[(a_{BA})\cdot(\beta^{\bullet}-c)-(a_{BA})^{2}\cdot(r^{\bullet})]$.
In particular, the entrant's fitness strictly decreases in the distance
between $a_{BA}$ and the Stackelberg strategy,\footnote{This strategy is the one chosen by the first-mover in the game where each player moves sequentially, with the second-mover observing the action of the first mover---in other words, in cases where one player can commit to an action and the other player chooses a best reply to that action.} $a_{\text{stack}}=\frac{\beta^{\bullet}-c}{2r^{\bullet}}$. 
\item Suppose the entrants have slope perception $\hat{r}$. Then in every
equilibrium zeitgeist, we have $a_{BA}=\frac{\beta^{\bullet}-c}{2\hat{r}+r^{\bullet}}$. 
\end{enumerate}
\end{prop}

We conclude with the following three observations, motivated by this result: 

\medskip 

\noindent \textbf{(1) Local vs. global mutations of the correctly specified model}.
Since $a_{BA}=\frac{\beta^{\bullet}-c}{2\hat{r}+r^{\bullet}}$, entrants
who correctly perceive $\hat{r}=r^{\bullet}$ will play $a_{BA}=\frac{\beta^{\bullet}-c}{3r^{\bullet}}$
and have the same fitness as the residents. As this is lower than
the Stackelberg strategy $a_{\text{stack}}=\frac{\beta^{\bullet}-c}{2r^{\bullet}}$,
the fitness of the entrants strictly decreases as $\hat{r}$ increases
above $r^{\bullet}$ and strictly increases as $\hat{r}$ decreases
below $r^{\bullet}$ up until $\hat{r}=r^{\bullet}/2$. Suppose we
consider the set of possible entrants that have a ``local'' mutation
relative to the correctly specified model: that is, entrants who have
a slope perception $\hat{r}$ with $|\hat{r}-r^{\bullet}|\le\delta$
for some small $\delta>0$ with $\delta<r^{\bullet}/2$. Among this
family of possible entrants, the entrant with the slope perception
$\hat{r}=r^{\bullet}-\delta$ has the highest fitness in equilibrium.
But, if we do not limit the entrants to only have these local mutations,
then the entrant who has the slope perception $\hat{r}=r^{\bullet}/2$
would have an even higher fitness. 

\medskip

\noindent \textbf{(2) Shifts vs. expansions of the correctly specified model}. The
model of the entrants who have a slope misperception $\hat{r}\ne r^{\bullet}$
can be viewed as a ``shift'' of the correctly specified model. Indeed,
the entrants have $\mathcal{F}_{B}=\{F_{\hat{r},\beta}:\beta\in\mathbb{R}\}$
while the residents have $\mathcal{F}_{A}=\{F_{r^{\bullet},\beta}:\beta\in\mathbb{R}\}$,
two disjoint sets of conjectures about how consequences are generated
in the stage game. \cite{FL_mutation} study the stability of models under local
mutations, but their notion of local mutation is that of a local \emph{expansion},
where the mutated model contains all the parameters of the resident
models and also some new parameters. In this Cournot example, if the
entrants simply entertain more possible values for the slope of the
demand curve than the residents, that is to say the entrants' model
has $\mathcal{F}_{B}=\{F_{r,\beta}:|r-r^{\bullet}|<\delta,\beta\in\mathbb{R}\}$,
then there is an EZ where the entrants and the residents come to the
same (correct) beliefs about both the slope and intercept of the demand
curve, play the same strategies, and have the same fitness. 

Going beyond this example, it is not difficult to see that in any
environment with strategic certainty, the correctly specified model
is never evolutionarily fragile against any expansions of it for the
same reason. If the resident model is evolutionarily fragile against
an entrant model, some of the feasible beliefs under the former must
be impossible under the latter. 

\medskip

\noindent \textbf{(3) An ``illusion of control'' entrant model that maximizes
entrant fitness}. There is another entrant model outside of the slope
misperception class discussed so far  that also maximizes
entrant fitness across all possible entrant models. This is a singleton
model with $\mathcal{F}_{B}=\{F^{*}\}$, where the conjectured distribution
of consequences $F^{*}(a_{i},a_{-i})$ only depends on $a_{i}$ and
not on $a_{-i}.$ In particular, $F^{*}(a_{i},a_{-i})$ specifies
that the first dimension of $y_{i}$ is always $a_{i}$, and the second
dimension is distributed according to $\beta^{\bullet}-r^{\bullet}(a_{i}+\frac{\beta^{\bullet}-c-r^{\bullet}a_{i}}{2r^{\bullet}})+\varepsilon$.
The conjecture $F^{*}$ features ``illusion of control,'' for it
stipulates that the distribution of market prices only depends on
$i$'s strategy and not on that of $-i$. In particular, it says that
whatever strategy the opponent actually chooses, the realized market
price distribution when $i$ chooses $a_{i}$ is the objective market
price distribution when $-i$ plays the rational best response against
it. It is easy to see that under the belief $F^{*}$, $i$'s strictly
dominant strategy is to choose the Stackelberg strategy $a_{BA}=a_{\text{stack}}=\frac{\beta^{\bullet}-c}{2r^{\bullet}}$
in every EZ. The rational residents must play the rational best response
against it, so the entrants obtain the Stackelberg payoff as their
fitness. In the next section, we construct a similar illusion of control
model for more general environments to find the highest possible fitness
among all entrants when the residents are correctly specified.

\section{Stability Implications of Belief Endogeneity}\label{sec:new_stability_phenomena}

In this section, we focus on environments with strategic certainty
to illustrate some stability phenomena that distinguish misspecified
inference from dogmatic beliefs in our framework. The main novelty
of our framework relative to past work on the indirect evolutionary
approach is that agents' beliefs about the game (and hence, subjective preferences)
are endogenously determined. We showcase some
of the unique implications of this belief endogeneity. 

Belief endogeneity adds new ways for biased individuals to develop
strategic commitments in games. First, unlike agents with fixed subjective
preferences, misspecified learners with a fixed model can develop situation-specific
commitments that are better tailored to the stage game. We show this
mechanism expands the scope for invading rational societies. Second,
misinference can induce different beliefs for a misspecified agent depending on who they  most frequently interact with. This leads
to new stability phenomena and adds nuance to extrapolations of the
welfare implications of a misspecified model across different societies,
relative to that of a distorted subjective preference.

\subsection{When Is Misinference Necessary to Defeat Rationality? }\label{subsec:When-Is-the}

Our first result characterizes when misspecified models can \emph{only}
invade a rational society when inference is possible. More precisely,
when does there exist a distribution over situations such that the
correctly specified model is not evolutionarily fragile against any
singleton model, but it is evolutionarily fragile against some models
with multiple parameters? 

When the stage game is fixed, the preference evolution literature
has long recognized that commitment to the game's Stackelberg strategy
can allow entrants to outperform rational residents (see, for example, Section 2.5 of \citet{robson2011evolutionary}).
In our setting, if there is only one situation and the highest symmetric
Nash equilibrium payoff is lower than the game's Stackelberg payoff,
it is straightforward to show that the correctly specified model is
evolutionarily fragile against any singleton model that misperceives
the Stackelberg strategy to be strictly dominant. Adherents of this
entrant model play the Stackelberg strategy against every opponent
and enjoy strictly higher fitness than the rational residents when
the entrant population share is close to zero. 

When there are multiple situations, the analogous conclusion that
the rational residents must have the Stackelberg payoff situation-by-situation
to avoid invasion by a misspecified entrant requires the existence
of an entrant who can behave differently in different situations,
since the Stackelberg strategy can vary by situation. The proof of
the first part of Theorem \ref{thm:theoryneeded} constructs such
an entrant, using the same ``illusion of control'' idea from Section
\ref{sec:example}. The entrant makes inferences among multiple parameters in their
model, where the different parameters are different illusion-of-control
beliefs meant to be adopted in different situations. Under the suitable
beliefs for situation $G$, the agent thinks their consequence is
solely controlled by their own action and views the Stackelberg strategy
for situation $G$ as strictly dominant. Also, the entrant's model
is constructed so that there is no equilibrium zeitgeist where entrants
adopt a belief meant for situation $G'$ in a different situation
$G\ne G'$, as inference from data would cause the entrants to revise
their beliefs in favor of a better-fitting parameter in such a scenario. 

The second part of Theorem \ref{thm:theoryneeded} characterizes distributions
over situations so that no singleton entrant model can invade the
correctly specified residents. In some environments with multiple
situations, singleton models that are unable to make inferences and
adapt to the different situations can nevertheless still obtain higher fitness
than the rational residents. In the Cournot duopoly setup from Section
\ref{sec:example}, for instance, we know that entrants with the slope misperception
$\hat{r}=r^{\bullet}/2$ can outperform the residents. But now suppose
there are two situations with true slope coefficients of $r^{\bullet}=1$
and $r^{\bullet}=1.001$, equally likely. Then the singleton entrant
model with the slope misperception $\hat{r}=1/2$ continues to outperform
the residents on average, across the two situations. The content of
the second part of Theorem \ref{thm:theoryneeded} is to provide a
condition that ensures the multiple situations are ``sufficiently
different'' strategically so that the correctly specified resident
model is not evolutionarily fragile against any singleton model. 

Theorem \ref{thm:theoryneeded} is stated for an environment with a finite strategy space. We require some
notation. For $F:\mathbb{A}^{2}\to\Delta(\mathbb{Y})$, let $U_{i}(a_{i},a_{-i},F)$
represent $i$'s expected payoff under the strategy profile $(a_{i},a_{-i})$
if consequences are generated by $F$, that is $U_{i}(a_{i},a_{-i},F):=\mathbb{E}_{y\sim F(a_{i},a_{-i})}[\pi(y)]$.
In each situation $G$, let $v_{G}^{\text{NE}}\in\mathbb{R}$ be the
highest symmetric Nash equilibrium payoff in $G$, when agents choose
strategies from $\mathbb{A}$. For each $a_{i}\in\mathbb{A}$, let
$\underline{\text{BR}}(a_{i},G)$ be a rational best response against
the strategy $a_{i}$ in situation $G,$ breaking ties \emph{against}
the user of $a_{i}$. Let $\bar{v}_{G}\in\mathbb{R}$ be the Stackelberg
equilibrium payoff in situation $G$, breaking ties against the Stackelberg
leader, i.e.,

\begin{equation}
\bar{v}_{G}:=\max_{a_{i}}U_{i}(a_{i},\underline{\text{BR}}(a_{i},G),F^{\bullet}(G)).\label{eq:stackleberg}
\end{equation}

\noindent Call the strategy $\bar{a}_{G}$ that maximizes Equation
(\ref{eq:stackleberg}) the Stackelberg strategy in situation $G$.
We assume the Stackelberg strategy is unique in each situation, and
furthermore that there is a unique rational best response to $\bar{a}_{G}$
in each situation $G',$ where possibly $G\ne G'$. Finally, for $b:\mathbb{A}\rightrightarrows\mathbb{A}$
a subjective best-response correspondence (which may be induced by
some belief about how consequences are distributed under different
strategy profiles in the stage game), let $v_{G}^{b}$ denote the
worst equilibrium payoff of an agent with the best-response correspondence
$b$ when she plays against a rational opponent in situation $G.$
That is, $v_{G}^{b}\in\mathbb{R}$ is $i$'s lowest payoff across
all strategy profiles $(a_{i},a_{-i})$ such that $a_{i}\in b(a_{-i})$
and $a_{-i}$ is a rational response to $a_{i}$ in situation $G.$\footnote{If no such profile exists, let $v_{G}^{b}=-\infty.$}

We impose two identifiability conditions: 
\begin{defn}
\emph{Situation identifiability} is satisfied if for every $a_{i},a_{-i}\in\mathbb{A}$
and $G\ne G',$ we have $F^{\bullet}(a_{i},a_{-i},G)\ne F^{\bullet}(a_{i},a_{-i},G').$
\emph{Stackelberg identifiability} is satisfied if whenever $G\ne G'$
and $a_{-i}$, $a_{-i}'$ are rational best responses to $\bar{a}_{G}$
in situations $G$ and $G'$, we have $F^{\bullet}(\bar{a}_{G},a_{-i},G)\ne F^{\bullet}(\bar{a}_{G},a_{-i}',G')$. 
\end{defn}
Under situation identifiability, a minimal correctly specified agent
can identify the true situation. Under Stackelberg identifiability,
playing the Stackelberg strategy $\bar{a}_{G}$ for any situation
$G$ generates consequence data that can statistically distinguish
whether the true situation is $G$ or not, provided the opponent chooses
a rational best response to the strategy for the true situation. 

The following result presents our characterization of when misinference
 is required for misspecified models to outperform rationality,
for some distribution over situations. The first part of the result
says, under identifiability assumptions and other regularity conditions,
the rational residents are always evolutionarily fragile against some
entrants unless they are already getting the Stackelberg payoff in
every situation. The second part of the result provides a condition
for the rational residents to not be evolutionarily fragile against
any singleton entrant. Whenever both conditions in Theorem \ref{thm:theoryneeded}
are satisfied, there is some distribution over situations so that
the minimal correctly specified model is evolutionarily fragile against
\emph{some} entrant model, but not evolutionarily fragile against
any \emph{singleton} entrant model. In these environments, the ability
to adapt preferences endogenously to the relevant situation (i.e.,
belief endogeneity) is a necessary condition for an invading entrant
to displace the rational resident. Hence, this result shows that entrants
with misspecified models cannot in general be represented simply as
entrants with fixed subjective preferences. 
\begin{thm}
\label{thm:theoryneeded} Suppose there are finitely many situations
and there is a symmetric Nash equilibrium in $\mathbb{A}\times\mathbb{A}$
for every situation $G$. 
\begin{enumerate}
\item If $v_{G}^{\text{NE}}<\bar{v}_{G}$ for some $G$, situation identifiability
and Stackelberg identifiability hold, and there are finitely many
strategies, then there exists a model $\hat{\Theta}$ such that the minimal 
correctly specified model is evolutionarily fragile against $\hat{\Theta}$
under any full-support distribution $q\in\Delta(\mathcal{G})$. 
\item If there is no point $(u_{G})_{G\in\mathcal{G}}$ in the convex hull
of $\{(v_{G}^{b})_{G\in\mathcal{G}}\mid b:\mathbb{A}\rightrightarrows\mathbb{A}\}$
with the property that $u_{G}\ge v_{G}^{\text{NE}}$ for every $G\in\mathcal{G},$
then there exists a full-support distribution $q\in\Delta(\mathcal{G})$
so that the minimal correctly specified model is not evolutionarily fragile
against any singleton model. 
\end{enumerate}
\end{thm}
One environment where $v_{G}^{\text{NE}}=\bar{v}_{G}$ for every situation
$G$ is when agents face decision problems --- that is, a player's
payoff in every situation $G$ is independent of the action of the
matched opponent. When all situations are decision problems, the condition
in the first part of Theorem \ref{thm:theoryneeded} is violated:
in fact, the correctly specified model is not evolutionarily fragile
against any other model, regardless of whether such invaders infer
from data. 

For the second part of the theorem, if there is no subjective best-response
correspondence $b$ such that $v_{G}^{b}\ge v_{G}^{\text{NE}}$ for
every $G\in\mathcal{G},$ then for every singleton entrant model there
exists some distribution over situations so that the correctly specified
model is not evolutionarily fragile against it. But the stronger condition
we impose ensures that there exists one distribution over situations
for which the correctly specified model is not evolutionarily fragile
against \emph{any} singleton entrant model. 

Next, we use a numerical example to illustrate how we might verify
the two sets of conditions from Theorem \ref{thm:theoryneeded}. 
\begin{example}
\label{exa:only_theory_invades} Suppose $\mathbb{A}=\{a_{1},a_{2},a_{3}\}$,
the consequences are $\text{\ensuremath{\mathbb{Y}}}=\{g,b\}$ with
$u(g)=1$ and $u(b)=0.$ Suppose there are two situations, $G_{A}$
and $G_{B}$, and the probability a given player obtains $g$ given
a strategy profile and situation is determined by the table below. 
\begin{center}
{\small{}%
\begin{tabular}{|c|c|c|c|}
\hline 
{\small$G_{A}$} & {\small$a_{1}$} & {\small$a_{2}$} & {\small$a_{3}$}\tabularnewline
\hline 
\hline 
{\small$a_{1}$} & {\small 0.1, 0.1} & {\small 0.1, 0.1} & {\small 0.1, 0.11}\tabularnewline
\hline 
{\small$a_{2}$} & {\small 0.1, 0.1} & {\small 0.3, 0.3} & {\small 0.1, 0.1}\tabularnewline
\hline 
{\small$a_{3}$} & {\small 0.11, 0.1} & {\small 0.1, 0.1} & {\small 0.2, 0.2}\tabularnewline
\hline 
\end{tabular}}{\small{} \qquad{}}{\small{}%
\begin{tabular}{|c|c|c|c|}
\hline 
{\small$G_{B}$} & {\small$a_{1}$} & {\small$a_{2}$} & {\small$a_{3}$}\tabularnewline
\hline 
\hline 
{\small$a_{1}$} & {\small 0.11, 0.11} & {\small 0.5, 0.5} & {\small 0.12, 0.4}\tabularnewline
\hline 
{\small$a_{2}$} & {\small 0.5, 0.5} & {\small 0.12, 0.12} & {\small 0.14, 0.55}\tabularnewline
\hline 
{\small$a_{3}$} & {\small 0.4, 0.12} & {\small 0.55, 0.14} & {\small 0.4, 0.4}\tabularnewline
\hline 
\end{tabular}}{\small\par}
\par\end{center}
It is easy to verify that the payoff-maximizing symmetric Nash equilibria
in $\mathbb{A}\times\mathbb{A}$ are $(a_{2},a_{2})$ for $G_{A}$
and $(a_{3},a_{3})$ for $G_{B}$, so $v_{G_{A}}^{\text{NE}}=0.3$
and $v_{G_{B}}^{\text{NE}}=0.4$. The unique Stackelberg strategy
in $G_{A}$ is $a_{2}$ and the unique Stackelberg strategy in $G_{B}$
is $a_{1}$, and we have $\bar{v}_{G_{A}}=0.3$ and $\bar{v}_{G_{B}}=0.5.$
There is a unique rational best response to every strategy in every
situation. 

Since $v_{G_{B}}^{\text{NE}}<\bar{v}_{G_{B}}$, the first set of conditions
of Theorem \ref{thm:theoryneeded} will be satisfied if we have situation
identifiability and Stackelberg identifiability. By inspection, the
probability of the $g$ outcome under every strategy profile differs
across the two situations, so situation identifiability holds. For
Stackelberg identifiability, note that the unique rational best response
to $a_{2}$ is $a_{2}$ in $G_{A}$ and $a_{3}$ in $G_{B}$, and
we have $F^{\bullet}(a_{2},a_{2},G_{A})\ne F^{\bullet}(a_{2},a_{3},G_{B}).$
The unique rational best response to $a_{1}$ is $a_{2}$ in $G_{B}$
and $a_{3}$ in $G_{A}$, and we have $F^{\bullet}(a_{1},a_{2},G_{B})\ne F^{\bullet}(a_{1},a_{3},G_{A})$.
So, Stackelberg identifiability also holds. 

To check the second set of conditions of Theorem \ref{thm:theoryneeded},
we consider three cases for the subjective best-response correspondence
$b$. 
\begin{itemize}
\item If $a_{1}\in b(a_{3})$, then since $a_{3}$ is a rational best response
to $a_{1}$ in $G_{A}$ and the highest possible payoff in $G_{B}$
is $0.55$, we get $v^{b}\le(0.1,0.55)$. 
\item If $a_{2}\in b(a_{3}),$ then since $a_{3}$ is a rational best response
to $a_{2}$ in $G_{B}$ and the highest possible payoff in $G_{A}$
is $0.3$, we get $v^{b}\le(0.3,0.14).$ 
\item If $a_{3}\in b(a_{3})$, then since $a_{3}$ is a rational best response
to $a_{3}$ in both $G_{A}$ and $G_{B}$, we get $v^{b}\le(0.2,0.4)$. 
\end{itemize}
These three cases are exhaustive since $b(a_{3})$ cannot be empty.
The half space in $\mathbb{R}^{2}$ below the line that runs through
$(0.1,0.55)$ and $(0.2,0.4)$ contains all three points $(0.1,0.55),$
$(0.3,0.14)$, and $(0.2,0.4).$ So, the convex hull of $\{(v_{G}^{b})_{G\in\mathcal{G}}\mid b:\mathbb{A}\rightrightarrows\mathbb{A}\}$
is contained in this half space. But we have $v^{\text{NE}}=(0.3,0.4)$,
which is outside of the half space. Thus, the second set of conditions
of Theorem \ref{thm:theoryneeded} are satisfied. 

We conclude, by Theorem \ref{thm:theoryneeded}, that there exists some
full-support distribution over the two situations $G_{A}$ and $G_{B}$
such that the correctly specified model is evolutionarily fragile
against some entrant model, but it is not evolutionarily fragile against
any singleton model. In fact, we can take this distribution to be
the one where the two situations are equally likely, and we can construct
the invading entrant model as one that features illusion of control.
This model has $\mathcal{F}=\{F_{A},F_{B}\},$ where both $F_{A}$
and $F_{B}$ stipulate that the consequence only depends on the agent's
own strategy and not on the opponent's strategy. Under $F_{A},$ $a_{1},a_{2}$,
and $a_{3}$ lead to consequence $g$ with probabilities 0.1, 0.3,
and 0.2 respectively (which are the probabilities of $g$ if opponent
plays a rational best response to these strategies in situation $G_{A}$).
Under $F_{B},$ playing $a_{1},a_{2}$, and $a_{3}$ lead to consequence
$g$ with probabilities 0.5, 0.14, and 0.4 respectively (which are
the probabilities of $g$ if opponent plays a rational best response
to these strategies in situation $G_{B}$).
\end{example}

\subsection{Stability Reversals}\label{subsec:Stability-Reversals}

We now highlight another consequence of the endogeneity of misspecified
beliefs: the potential for a greater indeterminacy in the emergence
of stable biases. For expositional simplicity, we assume that $|\mathcal{G}|=1$
throughout this section. We will refer to a model's \emph{conditional
fitness against group $g$}, i.e., the expected payoff of the model's
adherents in matches against group $g.$ 
\begin{defn}
Two models $\Theta_{A},\Theta_{B}$ exhibit \emph{stability reversal}
if (i) in every EZ with $(p_{A},p_{B})=(1,0),$ $\Theta_{A}$ has
strictly higher conditional fitness than $\Theta_{B}$ against group
A opponents and against group B opponents, but also (ii) in every
EZ with $(p_{A},p_{B})=(0,1),$ $\Theta_{B}$ has strictly higher
fitness than $\Theta_{A}$. 
\end{defn}
When $p_{B}=0$, how $\Theta_{A}$ performs against $\Theta_{B}$
does not actually affect group A's fitness. Condition (i) encodes
the strong requirement that $\Theta_{A}$ outperforms $\Theta_{B}$
even on the zero-probability event of being matched against a $\Theta_{B}$
opponent. A stability reversal occurs if this stronger requirement
holds (when $\Theta_{A}$ dominates in society), and yet $\Theta_{B}$
still strictly outperforms $\Theta_{A}$ if $\Theta_{B}$ starts from
a position of prominence.

We begin with two general results on when stability reversals \emph{cannot}
emerge. First, it cannot emerge without belief endogeneity:
\begin{prop}
\label{prop:no_reversal}Suppose $|\mathcal{G}|=1$. Two singleton
models (i.e., two subjective preferences in the stage game) cannot
exhibit stability reversal. 
\end{prop}
The reason is that for two singleton models, the conditional fitness
of group $g$ against group $g'$ does not depend on the relative
sizes of the groups. The subjective preference associated with a singleton
model never changes with the social composition, so a strategy profile
between groups $g$ and $g'$ that can be sustained in an EZ with
$(p_{A},p_{B})=(1,0)$ can also be sustained in an EZ with $(p_{A},p_{B})=(0,1)$. 

Stability reversals also cannot emerge in decision problems.
We show this by introducing a class of models where agents always
believe that strategic interactions do not matter:
\begin{defn}
A model $\Theta$ is \emph{strategically independent} if for all $\mu\in\Delta(\Theta)$,
$\underset{a_{i}\in\mathbb{A}}{\arg\max}\ U_{i}(a_{i},a_{-i};\mu)$
is the same for every $a_{-i}\in\mathbb{A}.$ 
\end{defn}
\noindent The adherents of a strategically independent model believe
that while an opponent's action may affect their utility, it does
not affect their best response. 
\begin{prop}
\label{prop:reversal_inference_channel}Suppose $|\mathcal{G}|=1$,
suppose $\Theta_{A},\Theta_{B}$ exhibit stability reversal and $\Theta_{A}$
is the correctly specified singleton model. Then, the beliefs that
the adherents of $\Theta_{B}$ hold in all EZs with $p=(1,0)$ and
the beliefs they hold in all EZs with $p=(0,1)$ form disjoint sets.
Also, $\Theta_{B}$ is not strategically independent. 
\end{prop}
\noindent The first claim of Proposition \ref{prop:reversal_inference_channel}
underscores that stability reversal requires inference---it cannot
happen if group B agents merely have a different subjective preference.
The second claim shows that stability reversal can only happen if
the misspecified agents respond differently to different rival play,
immediately implying they cannot emerge in decision problems. The
idea is that when the group B agents are prominent in the society,
their misperception that the stage game is a decision problem implies
that they will always choose the same strategy (say, $\hat{a}_{i}$)
against both group A and group B opponents. But this means their fitness
cannot be strictly higher than that of the rational group A agents,
who play a rational best response against $\hat{a}_{i}$ when they
match up against group B opponents. 

We now show by example that stability reversal can emerge with models
that allow for inference. Consider a two-player investment game where
player $i$ chooses an investment level $a_{i}\in\{1,2\}.$ A random
productivity level $P$ is realized according to $b^{\bullet}(a_{i}+a_{-i})+\epsilon$
where $\epsilon$ is a zero-mean noise term, $b^{\bullet}>0$. Player
$i$'s payoffs are $a_{i}\cdot P- (a_i - 1) \cdot c$. Consequences
are $y=(a_{i},a_{-i},P).$ We record the payoff matrix of this investment
game: 
\begin{center}
\begin{tabular}{|c|c|c|}
\hline 
\foreignlanguage{american}{} & 1 & 2\tabularnewline
\hline 
\hline 
1 & $2b^{\bullet},2b^{\bullet}$ & $3b^{\bullet},6b^{\bullet}-c$\tabularnewline
\hline 
2 & $6b^{\bullet}-c,3b^{\bullet}$ & $8b^{\bullet}-c,8b^{\bullet}-c$\tabularnewline
\hline 
\end{tabular}
\par\end{center}

\begin{condition} \label{cond:medium_cost}$5b^{\bullet}<c<6b^{\bullet}$.
\end{condition} In words, we assume that $a_{i}=1$ is a strictly
dominant strategy in the stage game, but the investment profile (2,2)
Pareto dominates the investment profile (1,1) (so that the corresponding
game is a prisoner's dilemma). Consider two models in the society.
Take $\Theta_{A}$ to be a correctly specified singleton (thus knowing
the true mapping from actions to payoffs), while $\Theta_{B}$ wrongly
stipulates $P=b(a_{i}+a_{-i})-m+\epsilon$, where $m>0$ is fixed,
while $b\in\mathbb{R}$ is a parameter that the adherents infer. We
impose a condition on $\Theta_{B}$, which holds whenever $m>0$ is
large enough: \begin{condition} \label{cond:large_misspec} $c<4b^{\bullet}+\frac{1}{3}m$
and $c<5b^{\bullet}+\frac{1}{4}m.$ \end{condition} We show that
in this example models $\Theta_{A}$ and $\Theta_{B}$ exhibit stability
reversal. 
\begin{example}
\label{exa:stability_reversal_example}In the investment game, under
Condition \ref{cond:medium_cost} and Condition \ref{cond:large_misspec},
$\Theta_{A}$ and $\Theta_{B}$ exhibit stability reversal. 
\end{example}
The idea is that the $\Theta_{B}$ adherents hold endogenous beliefs
about the value of $b$. They overestimate the complementarity of
investments, and this overestimation is more severe when they face
data generated from lower investment profiles. As a result, the match
between $\Theta_{A}$ and $\Theta_{B}$ plays out differently depending
on which model is resident: it results in the investment profile $(1,2)$
when $\Theta_{A}$ is resident, but results in $(1,1)$ when $\Theta_{B}$
is resident. (We relegate the formal argument to Appendix \ref{App:ExDetails}.)
Due to Propositions \ref{prop:no_reversal} and \ref{prop:reversal_inference_channel},
we conclude that this example is possible due to the non-trivial strategic
interactions and $\Theta_{B}$'s inference about $b$.

Stability reversals provide a clear demonstration of the endogeneity
of beliefs and hence the fluidity of conditional fitness in models
that permit inference. An entrant model may appear weak when present
in small proportions, doing worse than the resident model conditional
on every type of opponent. Yet, if the population share of the entrant
model reaches a critical mass, its adherents infer a more evolutionarily
advantageous model parameter based on their within-group interactions,
change their best-response correspondence, and hence outperform the
adherents of the resident model.

\section{Evolutionary Stability of Analogy Classes}\label{sec:ABEE}

We apply the stability notions introduced in this paper to study coarse
thinking in games. \cite{jehiel2005analogy} introduced analogy-based
expectation equilibrium (ABEE) in extensive-form games, where agents
group opponents' nodes into \emph{analogy classes} and only keep track
of aggregate statistics of opponents' average behavior within each
analogy class. An ABEE is a strategy profile where agents best respond
to the belief that at all nodes in every analogy class, opponents
behave according to the average behavior in the analogy class. The
ensuing literature typically treats analogy classes as exogenously
given, interpreted as arising from coarse feedback or agents' cognitive
limitations.\footnote{Section 6.2 of \cite{jehiel2005analogy} mentions that if players
could choose their own analogy classes, then the finest analogy classes
need not arise, but also says ``it is beyond the scope of this paper
to analyze the implications of this approach.'' In a different class
of games, \cite{jehiel1995limited} similarly observes that another
form of bounded rationality (having a limited forecast horizon about
opponent's play) can improve welfare.} We showcase the practical value of our approach by using the framework
from Section \ref{sec:Environment-and-Stability} to endogenize analogy
classes based on their objective expected payoffs in equilibrium.\footnote{Other approaches to endogenizing analogy classes are pursued in \cite{JehielMohlin2023,JehielWeber2023}.}

\subsection{Defining Stable Population Shares}\label{sect:StableShares}

In this section, we will focus on an environment where agents know
the stage game but may have misspecified beliefs about others' strategies. We will no longer work in the special case of strategic
certainty, and in fact we turn off the monitoring signals by assuming that $m_i$ is fully uninformative about the matched opponent's strategy $a_{-i}$. We will also be interested in stable population shares
in a society that contains positive fractions of both rational and
misspecified players. This is because the environment we analyze features a rational model and a misspecified model with neither model  being evolutionarily stable against the other (as we will see later in Proposition \ref{prop:abee}). 

We briefly introduce the following solution
concept.
\begin{defn}
\label{def:stability_interior}Call population share $(p,1-p)$ with
$p\in(0,1)$ a \emph{stable population share} if there is an EZ with
$(p,1-p)$ where both models have the same fitness, and there exists
$\bar{\epsilon}$ such that: 
\begin{enumerate}
\item For any $0<\epsilon<\bar{\epsilon}$, there is an EZ with population
share $(p+\epsilon,1-p-\epsilon)$ where $\Theta_{A}$ has strictly
lower fitness than $\Theta_{B}$ 
\item For any $0<\epsilon<\bar{\epsilon}$, there is an EZ with population
share $(p-\epsilon,1-p+\epsilon)$ where $\Theta_{A}$ has strictly
higher fitness than $\Theta_{B}$. 
\end{enumerate}
\end{defn}
\noindent Whereas Definition \ref{def:stability}'s stability notion
involves comparing the performance of the two models when one of them
is present in an arbitrarily small fraction, stability with an interior
population share as in Definition \ref{def:stability_interior} refers
to both models co-existing with equal fitness in a way that is robust
to local perturbations of population sizes.

Another difference between Definition   \ref{def:stability} and Definition \ref{def:stability_interior} is that the former requires a uniform welfare comparison across all EZs and the latter just requires a welfare comparison in one EZ. Indeed, we will select a particular focal EZ, because the environment has trivial EZs where misspecified agents always ``opt out'' of playing the game, receive no information about how others play, and hold beliefs about others' strategies that make opting out subjectively optimal. In such EZs, misspecified agents have the same fitness as the rational agents, but not for any interesting reasons that relate to their misspecified models. 

\subsection{Centipede Games and Analogy-Based Reasoning}

We now analyze analogy-based reasoning in the centipede game in Figure
\ref{fig:The-centipede-game} (there is only one situation, given
by the payoffs in this game). P1 and P2 take turns choosing Across
(A) or Drop (D). The non-terminal nodes are labeled $n^{k}$, $1\le k\le K$
where $K$ is an even number. P1 acts at odd nodes and P2 acts at
even nodes, where choosing Drop at $n^{k}$ leads to the terminal
node $z^{k}$. If Across is always chosen, then the terminal node
$z^{end}$ is reached. Every time a player $i$ chooses Across, the
sum of payoffs grows by $g>0.$ However, if the opponent chooses Drop next,
$i$'s payoff is $\ell>0$ smaller than $i$'s payoff had they chosen
Drop, with $\ell>g$. Thus, if $z^{end}$ is reached, both get $Kg/2;$
if $z^{k}$ is reached when $k$ is odd, both players obtain $\frac{g(k-1)}{2}$;
and if  $z^{k}$ is reached when $k$ is even, P1 obtains $\frac{k-2}{2}g-\ell$,
and P2 obtains $\frac{k}{2}g+\ell$.

\begin{figure}[h]
\begin{centering}
\includegraphics[scale=0.35]{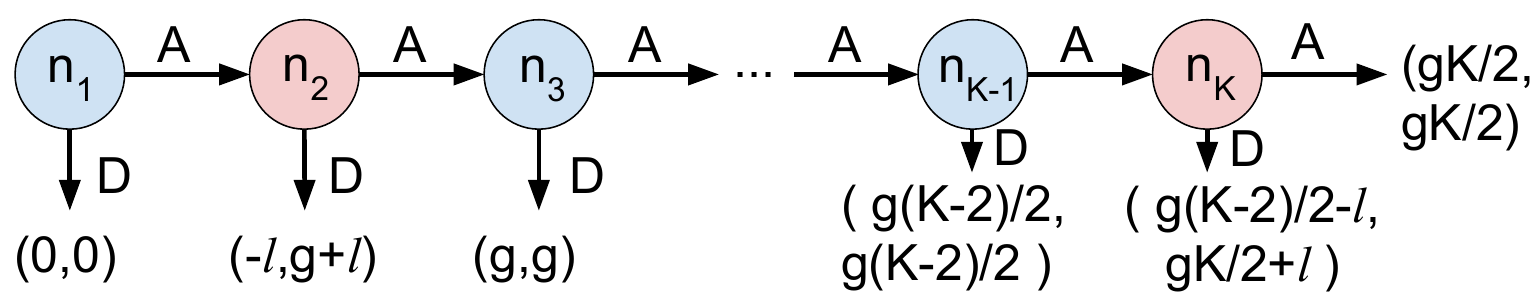}
\par\end{centering}
\vspace{-0bp}

\caption{The centipede game. P1 (blue) and P2 (red) alternate in choosing
Across (A) or Drop (D). Payoff profiles are shown at the terminal
nodes.}\label{fig:The-centipede-game}
\end{figure}

While this is an asymmetric stage game, we study a symmetrized version
where two matched agents are randomly assigned into the roles of P1
and P2. Let $\mathbb{A}=\{(d^{k})_{k=1}^{K}\in[0,1]^{K}\}$, so each
strategy is characterized by the probabilities of playing Drop at
various nodes in the game tree. When assigned into the role of P1,
the strategy $(d^{k})$ plays Drop with probabilities $d^{1},d^{3},...,d^{K-1}$
at nodes $n^{1},n^{3},...n^{K-1}$. When assigned into the role of
P2, it plays Drop with probabilities $d^{2},d^{4},...,d^{K}$ at nodes
$n^{2},n^{4},...n^{K}$. The set of consequences is $\mathbb{Y}=\{1,2\}\times(\{z_{k}:1\le k\le K\}\cup\{z_{end}\})$,
where the first dimension of the consequence returns the player role
that the agent was assigned into, and the second dimension returns
the terminal node reached. Let $F^{\bullet}:\mathbb{A}^{2}\to\Delta(\mathbb{Y})$
be the objective distribution over consequences.

All agents know the game tree (i.e., $F^{\bullet}$), but some might
adhere to a model which mistakenly assumes that their opponent plays
Drop with the same probabilities at all of their nodes. Formally,
define the restricted space of strategies $\mathbb{A}^{An}:=\{(d^{k})\in[0,1]^{K}:d^{k}=d^{k'}\text{ if }k\equiv k^{'}\text{(mod 2)}\}\subseteq\mathbb{A}$.
The correctly specified model is $\Theta^{\bullet}:=\mathbb{A}\times\mathbb{A}\times\{F^{\bullet}\}.$
The misspecified model of interest is $\Theta^{An}:=\mathbb{A}^{An}\times\mathbb{A}^{An}\times\{F^{\bullet}\}$,
reflecting a dogmatic belief that opponents play the same mixed action
at all nodes in the analogy class. We emphasize these restriction
on strategies only exists in the subjective beliefs of the model $\Theta^{An}$
adherents. All agents, regardless of their model, actually have the
strategy space $\mathbb{A}$.

\subsection{Results}

The next proposition provides a justification for why we might expect
agents with coarse analogy classes given by $\mathbb{A}^{An}$ to
persist in the society. 
\begin{prop}
\label{prop:abee}Suppose $K\ge4$ and $g>\frac{2}{K-2}\ell$. The
correctly specified  model $\Theta^{\bullet}$ is evolutionarily stable
against itself, but it is not evolutionarily stable against the misspecified
model $\Theta^{An}.$ Also, $\Theta^{An}$ is not evolutionarily stable
against $\Theta^{\bullet}$. 
\end{prop}
Thus, the correctly specified model is not evolutionarily stable against
a coarse reasoner. Here, the conditional fitness of $\Theta^{An}$
against both $\Theta^{\bullet}$ and $\Theta^{An}$ can strictly improve
on the correctly specified residents' equilibrium fitness. This is
because the matches between two adherents of $\Theta^{\bullet}$ must
result in Dropping at the first move in equilibrium, while matches
where at least one player is an adherent of $\Theta^{An}$ either
lead to the same outcome or lead to a Pareto dominating payoff profile
as the misspecified agent misperceives the opponent's continuation
probability and thus chooses Across at almost all of the decision
nodes.

However, $\Theta^{An}$ is not evolutionarily stable against $\Theta^{\bullet}$
either. The correctly specified agents can exploit the analogy reasoners'
mistake and receive higher payoffs in matches against them than the
misspecified agents receive in matches against each other. Hence,
no homogeneous population can be stable, as the resident model would
have lower fitness than the entrant model in equilibrium. Thus we
determine stable shares as defined in Section \ref{sect:StableShares},
focusing on the EZs where Across is played as often as possible.

Suppose $K\ge4$ and $g>\frac{2}{K-2}\ell$. Consider the\emph{ maximal
continuation EZ}: (1) misspecified agents always play Across except
at node $K$ where they choose Drop, and (2) correctly specified agents
(i) when matched with misspecified agents, play Drop at nodes $K-1$
and $K$ and Across otherwise, and (ii) when matched with correctly
specified agents, always play Drop. We verify this indeed forms an
EZ.
\begin{prop}
\textup{\label{prop:stable_pop_share} }\emph{Suppose $K\ge4$ and
$g>\frac{2}{K-2}\ell$.}\textup{\emph{ }}\textup{The only stable population
share $(p_{A}^{*},p_{B}^{*})$ supported by the maximal continuation
EZ described above is $p_{B}^{*}=1-\frac{\ell}{g(K-2)}$. We have
$p_{B}^{*}$ is strictly increasing in $g$ and $K$, and strictly
decreasing in $\ell.$} 
\end{prop}
Intuitively, $p_{B}^{*}$ reflects the fraction of society expected
to be analogy reasoners if long-run population changes are determined
by fitness. Under the maintained assumption $g>\frac{2}{K-2}\ell,$
the stable population share of misspecified agents is strictly more
than 50\%, and the share grows with more periods and a larger increase
in payoffs from continuation. The main intuition is that the misspecified
model has a higher conditional fitness than the rational model against
rational opponents. The former leads to many periods of continuation
and a high payoff for the biased agent when the rational agent eventually
drops, but the latter leads to 0 payoff from immediate dropping. On
the other hand, the misspecified model has a lower conditional fitness
than the rational model against misspecified opponents. For the two
groups to have the same expected fitness, there must be fewer rational
opponents (i.e., a smaller stable population share $p_{A}^{*}$) when
$g$ and $K$ are higher.

Note that, when payoffs are specified as above, two successive periods
of continuation lead to a strict Pareto improvement in payoffs. Consider
instead the so-called ``dollar game'' \cite{reny1993common} in
Figure \ref{fig:The_dollar_game}, a variant with a more ``competitive''
payoff structure, where an agent always gets zero when the opponent
plays Drop, at all parts of the game tree. Assume total payoff increases
by 1 in each round. If the first player stops immediately, payoffs
are (1, 0). If the second player continues at the final node $n^{K}$,
payoffs are $(K+2,0).$

\begin{figure}[h]
\begin{centering}
\includegraphics[scale=0.35]{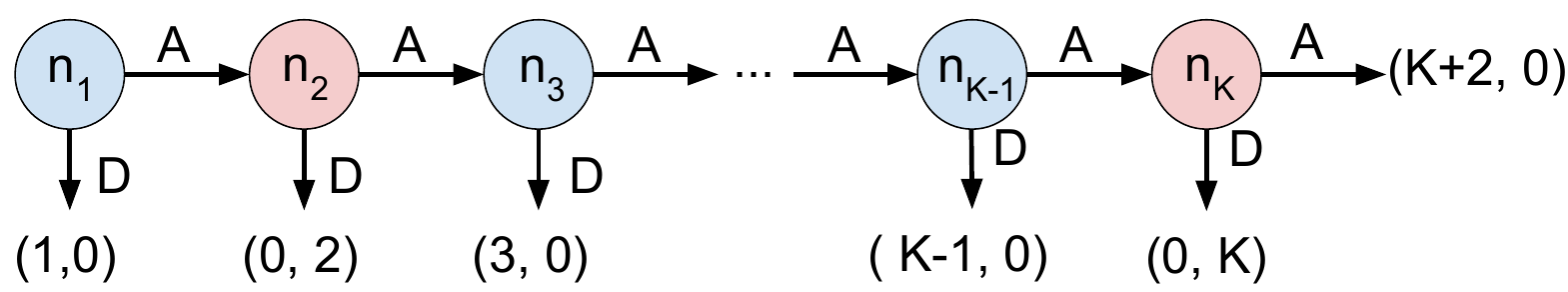}
\par\end{centering}
\vspace{-10bp}

\caption{The dollar game. P1 (blue) and P2 (red) alternate in choosing
Across (A) or Drop (D). Payoff profiles are shown at the terminal
nodes.}\label{fig:The_dollar_game}
\end{figure}

\begin{prop}
\textup{\label{prop:stable_pop_share_dollar} For every population
size $(p,1-p)$ with $p\in[0,1],$ }\emph{the maximal continuation
EZ}\textup{ is an EZ where the fitness of $\Theta^{\bullet}$ is strictly
higher than that of $\Theta^{An}$.} 
\end{prop}
While maximal continuation remains an EZ, the rational model strictly
outperforms the misspecified model for all population shares. Provided
the maximal continuation EZ remains focal, we would expect no analogy
reasoners in the long run with this stage game. Intuitively, the payoffs
imply one player can only do better \emph{at the expense} of the opponent.
This implies the less cooperative strategy will be selected.

In a recent survey, \cite{jehiel2020survey} points out that the misspecified
Bayesian learning approach to analogy classes should aim for ``a
better understanding of how the subjective theories considered by
the players may be shaped by the objective characteristics of the
environment.''\footnote{\cite{jehiel2020survey} interprets ABEEs as players adopting the
``simplest'' explanations of observed aggregate statistics of play
with coarse feedback. An objectively coarse feedback structure can
lead agents to adopt the subjective belief that others behave in the
same way in all contingencies in the same coarse analogy class.} Taken together, our analysis in this section provides predictions
regarding when coarse reasoning should be more prevalent, specifically
when the payoff structure is ``less competitive.'' When this is
indeed the case, the bias becomes more prevalent with a longer horizon
and with faster payoff growth.
\section{\label{sec:Concluding-Discussion}Concluding Discussion}

We have introduced an evolutionary approach to predict the persistence  of misspecified models under Bayesian learning. We have emphasized the implications and significance of belief endogeneity for evolutionary stability and the viability of models. Our contributions are twofold. First, we show that belief endogeneity may confer strategic benefits in cases where dogmatic beliefs do not. This is because endogenous beliefs enable flexible commitments that are tailored to the realized situation. Second, we show that the endogeneity of  misspecified beliefs makes it difficult to extrapolate the performance of a fixed bias across environments. More broadly, we hope to have shown that incorporating inference enables the evolutionary approach to speak to new applications and patterns. 

We acknowledge that our framework does not account for which errors
appear in the first place. It is plausible that some first-stage filter prevents certain obvious misspecifications from ever reaching the stage that we study in the evolutionary framework. For this reason, the applications we focused on reflected misspecifications that seem psychologically plausible.

We have used an otherwise off-the-shelf framework to describe the selection of specifications. The goal of this paper is
not to identify suitable definitions of fitness to justify particular
errors (which is the focus for many of the papers that \citet{robson2011evolutionary}
survey). Rather, our goal has been to determine what evolutionary forces would suggest about the persistence of misspecified models. We have therefore focused more on the implications of belief endogeneity in an otherwise standard evolutionary setup.

\vspace{-15bp}

\begin{singlespace}
{\small{}{}{} \bibliographystyle{ecta}
\bibliography{misspec_and_welfare}
}{\small\par}
\end{singlespace}

\vspace{-10bp}

\newpage

\appendix

\begin{center}
\textbf{\Large{}{}{}Appendix}{\Large\par}
\par\end{center}

\vspace{-30bp}

\section{\label{sec:Proofs}Omitted Proofs from the Main Text}

\global\long\def\thedefn{A.\arabic{defn}}%
\global\long\def\thecor{A.\arabic{cor}}%
\global\long\def\theprop{A.\arabic{prop}}%
\global\long\def\thelem{A.\arabic{lem}}%
\global\long\def\theclaim{A.\arabic{claim}}%
\global\long\def\theassumption{A.\arabic{assumption}}%
 \setcounter{prop}{0} \setcounter{lem}{0} \setcounter{defn}{0}
\setcounter{assumption}{0}

\subsection{Proof of Proposition \ref{prop:cournot_example}}
\begin{proof}
We first note that for an agent $i$ who believes in the parameters
$r>0$ and $\beta$, the subjective expected utility from the strategy
profile $(a_{i},a_{-i})$ is $a_{i}\cdot(\beta-r(a_{i}+a_{-i})-c)$
(since $\varepsilon$ is mean zero). The second derivative in $a_{i}$
is $-2r<0$, so the maximizer is characterized by the first-order
condition. Taking FOC, we get the subjective best response $a_{i}=\frac{\beta-c-ra_{-i}}{2r}$. 

We also note that the correctly specified residents must have an equilibrium
belief that assigns probability 1 to $\beta=\beta^{\bullet}$ in every
EZ. This is because in an EZ where the residents play the strategy
$a_{AA}$, $\beta=\beta^{\bullet}$ has zero KL divergence whereas
any other value of $\beta$ has strictly positive KL divergence. 

Now we prove the three parts of the proposition. 

Part 1: The only Nash equilibrium of the game is for both players
to choose $\frac{\beta^{\bullet}-c}{3r^{\bullet}}$ (since the rational
best response function is linear). So, the fitness of the residents
is given by $\frac{\beta^{\bullet}-c}{3r^{\bullet}}\cdot(\beta^{\bullet}-r^{\bullet}\cdot2\cdot\frac{\beta^{\bullet}-c}{3r^{\bullet}}-c)$,
which simplifies to $\frac{(\beta^{\bullet}-c)^{2}}{9r^{\bullet}}.$ 

Part 2: In an equilibrium zeitgeist where the entrants use the strategy
$a_{BA}$, the correctly specified residents (who have correct beliefs
about all the parameters in equilibrium) best respond with $a_{AB}=\frac{\beta^{\bullet}-c-r^{\bullet}a_{BA}}{2r^{\bullet}}$.
So, the fitness of the entrant is given by $a_{BA}\cdot(\beta^{\bullet}-r^{\bullet}(a_{BA}+\frac{\beta^{\bullet}-c-r^{\bullet}a_{BA}}{2r^{\bullet}})-c).$
Simplifying we get $\frac{1}{2}[(a_{BA})\cdot(\beta^{\bullet}-c)-(a_{BA})^{2}\cdot(r^{\bullet})].$
This expression is quadratic in $a_{BA}$ and it must be maximized
at the Stackelberg strategy of the game, hence the fitness of the
entrants is strictly decreasing in the distance between $a_{BA}$
and the Stackelberg strategy. The Stackelberg strategy is found by
taking the first-order condition of the expression $\frac{1}{2}[(a_{BA})\cdot(\beta^{\bullet}-c)-(a_{BA})^{2}\cdot(r^{\bullet})]$,
which gives $a_{\text{stack}}=\frac{\beta^{\bullet}-c}{2r^{\bullet}}$. 

Part 3: In an EZ where entrants with slope misperception $\hat{r}$
play $a_{i}$ against residents who play $a_{-i}$, the distribution
of consequences has zero KL divergence only for the parameter $\hat{\beta}$
that solves $\beta^{\bullet}-r^{\bullet}A=\hat{\beta}-\hat{r}A$,
so entrants must infer $\hat{\beta}=\beta^{\bullet}+(a_{i}+a_{-i})(\hat{r}-r^{\bullet})$.
But if the entrants play $a_{i}$, in equilibrium the residents must
play the rational best response against it, which is $a_{-i}=\frac{\beta^{\bullet}-c-r^{\bullet}a_{i}}{2r^{\bullet}}$.
So the entrants infer $\hat{\beta}=\beta^{\bullet}+(a_{i}+\frac{\beta^{\bullet}-c-r^{\bullet}a_{i}}{2r^{\bullet}})(\hat{r}-r^{\bullet})$
in an EZ where they choose $a_{BA}=a_{i}$. Under the beliefs $(\hat{\beta},\hat{r}),$
the entrants' subjective best response to $\frac{\beta^{\bullet}-c-r^{\bullet}a_{i}}{2r^{\bullet}}$
is $\frac{\hat{\beta}-\hat{r}\frac{\beta^{\bullet}-c-r^{\bullet}a_{i}}{2r^{\bullet}}-c}{2\hat{r}}$,
and so we must have $a_{i}=\frac{\hat{\beta}-\hat{r}\frac{\beta^{\bullet}-c-r^{\bullet}a_{i}}{2r^{\bullet}}-c}{2\hat{r}}$.
Making the substitution that $\hat{\beta}=\beta^{\bullet}+(a_{i}+\frac{\beta^{\bullet}-c-r^{\bullet}a_{i}}{2r^{\bullet}})(\hat{r}-r^{\bullet})$
on the right-hand side, we get $a_{i}=\frac{\beta^{\bullet}+(a_{i}+\frac{\beta^{\bullet}-c-r^{\bullet}a_{i}}{2r^{\bullet}})\cdot(\hat{r}-r^{\bullet})-\hat{r}\frac{\beta^{\bullet}-c-r^{\bullet}a_{i}}{2r^{\bullet}}-c}{2\hat{r}}$.
Simplifying this linear equation in $a_{i}$ gives us the unique solution
$a_{i}=\frac{\beta^{\bullet}-c}{2\hat{r}+r^{\bullet}}$, which is
the unique EZ value of $a_{BA}$ when entrants have the slope misperception
$\hat{r}$. 
\end{proof}
\subsection{Proof of Theorem \ref{thm:theoryneeded}}

\textit{Part 1:} Suppose the hypotheses hold and let us construct
the misspecified model $\hat{\Theta}=\{F_{G}:G\in\mathcal{G}\}.$
Towards defining  the parameter $F_{G}$ for each situation $G$,  first consider $\tilde{F}_{G}$ where
$\tilde{F}_{G}(a_{i},a_{-i}):=F^{\bullet}(a_{i},\underline{\text{BR}}(a_{i},G),G)$
for every $a_{-i}\in\mathbb{A}$. Now for each $(a_{i},a_{-i},G)\in\mathbb{A}\times\mathbb{A}\times\mathcal{G}$,
define the full-support distribution $F_{G}(a_{i},a_{-i})\in\Delta(\mathbb{Y})$
as a sufficiently small perturbation of  $\tilde{F}_{G}(a_{i},a_{-i})$,
such that for every $a_{i},a_{-i}\in\mathbb{A}$ and every $G\in\mathcal{G}$,
$\min_{\hat{G}\in\mathcal{G}}KL(F^{\bullet}(a_{i},a_{-i},G)\parallel F_{\hat{G}}(a_{i},a_{-i}))$
has a unique solution. This can be done because there are finitely
many strategies and situations.

Consider any  EZ $\mathfrak{Z}$ where the resident  is the minimal correctly specified model and the entrant  is 
$\hat{\Theta}$. By situation identifiability and because we are in an environment of strategic certainty,
in $\mathfrak{Z}$ the correctly specified residents must believe
in the true $F^{\bullet}(\cdot,\cdot,G)$ in every situation $G$.
When the fraction of entrants $\epsilon>0$ is sufficiently small, the entrants cannot hold a mixed belief in any situation $G$, by the
construction of the parameters in $\hat{\Theta}$ to rule out ties in
KL divergence if entrants only use the consequences in their matches against the residents to make inferences. We show further that entrants must believe in $F_{G}$
in situation $G$ for $\epsilon$ small enough. This is because if they instead believed in $F_{G'}$
for some $G'\ne G$, then they must play $\bar{a}_{G'}$ as the Stackelberg
strategy is assumed to be unique. Let $a_{-i}$ be the rational best
response to $\bar{a}_{G'}$ in situation $G$ and $a_{-i}'$ be the
rational best response to $\bar{a}_{G'}$ in situation $G',$ both
unique by assumption. In their matches against the residents, the entrants' expected distribution of consequences
$F_{G'}(\bar{a}_{G'},a_{-i})$ is a perturbed version of $F^{\bullet}(\bar{a}_{G'},a_{-i}',G')$,
while the true distribution of consequences $F^{\bullet}(\bar{a}_{G'},a_{-i},G)$
is a perturbed version of $F_{G}(\bar{a}_{G'},a_{-i})$. We have $F^{\bullet}(\bar{a}_{G'},a_{-i}',G')\ne F^{\bullet}(\bar{a}_{G'},a_{-i},G)$
by Stackelberg identifiability, so $KL(F^{\bullet}(\bar{a}_{G'},a_{-i},G)\parallel F_{G}(\bar{a}_{G'},a_{-i}))<KL(F^{\bullet}(\bar{a}_{G'},a_{-i},G)\parallel F_{G'}(\bar{a}_{G'},a_{-i}))$
when the perturbations are sufficiently small. When $\epsilon>0$ is small enough, this contradicts the
entrants believing in $F_{G'}$ in situation $G$ as the parameter $F_{G}$
generates smaller weighted KL divergence across all of the entrant's data (since data from matches against entrants get weighted by $\epsilon$ and the full-support nature of all processes in the model implies that KL divergence of the data from such matches is bounded). So the entrants get the Stackelberg
payoff in each situation when playing the resident, which means they have higher fitness than
the residents in every  EZ for $\epsilon$ small enough, since $\bar{v}_{G}>v_{G}^{\text{NE}}$
for at least one situation and $q$ has full support. Finally, there
exists at least one  EZ: for $\epsilon>0$ small enough, it is an  EZ for the residents to believe
in $F^{\bullet}(\cdot,\cdot,G)$ in every situation $G$, to play
the symmetric Nash profile that results in $v_{G}^{\text{NE}}$ when
matched with other residents (this profile exists by hypothesis of
the theorem), and for the entrants to believe in $F_{G}$ and play
$(\bar{a}_{G},\underline{\text{BR}}(\bar{a}_{G},G))$ in matches against
residents in situation $G.$

\textit{Part 2:} Let $\mathcal{V}$ be the convex hull of $\{(v_{G}^{b})_{G\in\mathcal{G}}\mid b:\mathbb{A}\rightrightarrows\mathbb{A}\}$,
and let $\mathcal{U}=\{(u_{G})_{G\in\mathcal{G}}:u_{G}\le v_{G}\text{ for all }G\text{ for some }v\in\mathcal{V}\}.$
Note $\mathcal{U}$ is closed and convex (since $\mathcal{V}$ is
convex). By hypothesis, $v^{\text{NE}}$ is not in the interior or
on the boundary of $\mathcal{U}.$ So by the separating hyperplane
theorem, there exists a real number $c$ and a vector $q\in\mathbb{R}^{|\mathcal{G}|}$ with
$q_{G}\ne0$ for every $G,$ so that $q\cdot v^{\text{NE}}>c>q\cdot u$
for every $u\in\mathcal{U}.$ Furthermore, $q_{G}\ge0$ for every
$G.$ This is because if $q_{G'}<0$ for some $G',$ then since $\mathcal{U}$
contains vectors with arbitrarily negative values in the $G'$ dimension,
we cannot have $q\cdot v^{\text{NE}}\ge q\cdot u$ for every $u\in\mathcal{U}.$
We may then without loss view $q$ as a distribution on $\mathcal{G}$. In fact, we can take $q$ to be full support.  To see this,  note that since $|\mathcal{G}| < \infty$ and $\mathcal{U}$ is convex, we have
\begin{equation*} 
\lim_{\varepsilon \rightarrow 0} \max_{v \in \mathcal{U}} \left[(1- \varepsilon)q + \frac{\varepsilon}{|\mathcal{G}|}(1,1, \ldots,1) \right] \cdot v =  \max_{v \in \mathcal{U}} q \cdot v,
\end{equation*} 

\noindent by continuity of the support function of convex sets in $\mathbb{R}^{n}$ (given that the support function on $\mathcal{U}$ is bounded for all $q \geq 0$, since $v_{G}^{b}$ is bounded above for every $b$ and every $G$). Thus, setting $\tilde{q}(\varepsilon) = (1- \varepsilon)q + \frac{\varepsilon}{|\mathcal{G}|}(1,1,\ldots, 1)$, we have $\tilde{q}(\varepsilon)$ is a full support distribution with $\tilde{q}(\varepsilon) \cdot v^{NE} > c > \tilde{q}(\varepsilon) \cdot u$ whenever $\varepsilon$ is sufficiently small, since we have that these inequalities hold in the limit.

Now consider any singleton model $\mathcal{F}=\{F\}$, and let $b:\mathbb{A}\rightrightarrows\mathbb{A}$
be the subjective best-response correspondence that $F$ induces. If
$v_{G}^{b}\ne-\infty$ for every $G$, then, for each $G$ we can
find a strategy profile $(a_{i}^{G},a_{-i}^{G})$ where $a_{i}^{G}\in b(a_{-i}^{G}),$
$a_{-i}^{G}$ is a rational best response to $a_{i}^{G}$ in situation
$G$, and the strategy pair gives payoff $v_{G}^{b}$ to the first
player. For any population shares of the two models, there is an  EZ where the resident correctly specified agents
get $v_{G}^{\text{NE}}$ in situation $G$ when playing against each other, and the entrants with model
$\Theta$ play $(a_{i},a_{-i})$ in matches against the residents
and get utility $v_{G}^{b}$ in the same situation. Under the distribution
of situations $q$, as the fraction of the entrants approaches 0, the residents' fitness approaches $q\cdot v^{\text{NE}}$
while that of the entrants approaches $q\cdot v^{b}$, and the former is strictly
larger by construction of $q$ since $v^{b}\in\mathcal{U}$. This
EZ shows the correctly specified model is not evolutionarily fragile
against $\{F\}.$ Otherwise, if we have that $v_{G}^{b}=-\infty$
for some $G,$ then there are no  EZs, so the correctly specified
model is not evolutionarily fragile against $\{F\}$ by the emptiness
of the set of  EZs.

\subsection{Proof of Proposition \ref{prop:no_reversal}}
\begin{proof}
Let two singleton models $\Theta_{A},\Theta_{B}$ be given. By contradiction,
suppose they exhibit stability reversal. Let $\mathfrak{Z}=(\mu_{A},\mu_{B},p=(0,1),(a))$
be any EZ where $\Theta_{B}$ is resident. By the definition of EZ,
$\mathfrak{Z}^{'}=(\mu_{A},\mu_{B},p=(1,0),(a))$
is also an EZ where $\Theta_{A}$ is resident. Let $u_{g,g^{'}}$
be model $\Theta_{g}$'s conditional fitness against group $g^{'}$
in the EZ $\mathfrak{Z}^{'}$. Part (i) of the definition of stability
reversal requires that $u_{AA}>u_{BA}$ and $u_{AB}>u_{BB}$. These
conditional fitness levels remain the same in $\mathfrak{Z}$. This
means the fitness of $\Theta_{A}$ is strictly higher than that of
$\Theta_{B}$ in $\mathfrak{Z}$, a contradiction.
\end{proof}

\subsection{Proof of Proposition \ref{prop:reversal_inference_channel}}
\begin{proof}
To show the first claim, suppose $\mathfrak{Z}=(\mu_{A},\mu_{B},p=(1,0),(a_{AA},a_{AB},a_{BA},a_{BB}))$
is an EZ, and $\mathfrak{\tilde{Z}}=(\mu_{A},\mu_{B},p=(0,1),(\tilde{a}_{AA},\tilde{a}_{AB},\tilde{a}_{BA},\tilde{a}_{BB}))$
is another EZ where the adherents of $\Theta_{B}$ hold the same belief
$\mu_{B}$ (group A's belief cannot change as $\Theta_{A}$ is the
correctly specified singleton model). By the optimality of behavior
in $\mathfrak{Z}$, $a_{BA}$ best responds to $a_{AB}$ under the
belief $\mu_{B}$, and $a_{AB}$ best responds to $a_{BA}$ under
the belief $\mu_{A}$, therefore $\mathfrak{\tilde{Z}}^{'}=(\mu_{A},\mu_{B},p=(0,1),(\tilde{a}_{AA},a_{AB},a_{BA},\tilde{a}_{BB}))$
is another EZ. This holds because the distributions of observations
for the adherents of $\Theta_{B}$ are identical in $\mathfrak{\tilde{Z}}$
and $\mathfrak{\tilde{Z}}^{'}$, since they only face data generated
from the profile $(\tilde{a}_{BB},\tilde{a}_{BB}).$ At the same time,
since $\tilde{a}_{BB}$ best responds to itself under the belief $\mu_{B},$
we have that $\mathfrak{Z^{'}}=(\mu_{A},\mu_{B},p=(1,0),(a_{AA},a_{AB},a_{BA},\tilde{a}_{BB}))$
is an EZ. Part (i) of the definition of stability reversal applied
to $\mathfrak{Z^{'}}$ requires that $U^{\bullet}(a_{AB},a_{BA})>U^{\bullet}(\tilde{a}_{BB},\tilde{a}_{BB})$
(where $U^{\bullet}$ is the objective expected payoffs), but part
(ii) of the same definition applied to $\mathfrak{\tilde{Z}}^{'}$
requires $U^{\bullet}(\tilde{a}_{BB},\tilde{a}_{BB})\ge U^{\bullet}(a_{AB},a_{BA}),$
a contradiction.

To show the second claim, by way of contradiction suppose $\Theta_{B}$
is strategically independent and $\mathfrak{Z}=(\mu_{A},\mu_{B},p=(0,1),(a_{AA},a_{AB},a_{BA},a_{BB}))$
is an EZ. By strategic independence, the adherents of $\Theta_{B}$
find it optimal to play $a_{BB}$ against any opponent strategy under
the belief $\mu_{B}$. So, there exists another EZ of the form $\mathfrak{Z^{'}}=(\mu_{A},\mu_{B},p=(0,1),(a_{AA},a_{AB}^{'},a_{BB},a_{BB}))$,
where $a_{AB}^{'}$ is an objective best response to $a_{BB}$. The
belief $\mu_{B}$ is sustained because in both $\mathfrak{Z}$ and
$\mathfrak{Z^{'}}$, the adherents of $\Theta_{B}$ have the same
data: from the strategy profile $(a_{BB},a_{BB}).$ In $\mathfrak{Z^{'}}$,
$\Theta_{A}$'s fitness is $U^{\bullet}(a_{AB}^{'},a_{BB})$ and
$\Theta_{B}$'s fitness is $U^{\bullet}(a_{BB},a_{BB}).$ We have
$U^{\bullet}(a_{AB}^{'},a_{BB})\ge U^{\bullet}(a_{BB},a_{BB})$ since
$a_{AB}^{'}$ is an objective best response to $a_{BB},$ contradicting
the definition of stability reversal.
\end{proof}

\subsection{Details Behind Example \ref{exa:stability_reversal_example}} \label{App:ExDetails}

Let $b^{*}(a_{i},a_{-i})$ solve  
\[ \min_{b\in\mathbb{R}}D_{KL}(F^{\bullet}(a_{i},a_{-i})\parallel\hat{F}(a_{i},a_{-i};b,m))), \]
where $F^{\bullet}(a_{i},a_{-i})$ is the objective distribution over
consequences under the investment profile $(a_{i},a_{-i}),$ and $\hat{F}(a_{i},a_{-i};b,m)$
is the distribution under the same investment profile if
 productivity is given by $P=b(x_{i}+x_{-i})-m+\epsilon$. We
find that $b^{*}(a_{i},a_{-i})=b^{\bullet}+\frac{m}{a_{i}+a_{-i}}$. It is clear that $D_{KL}(F^{\bullet}(a_{i},a_{-i})\parallel\hat{F}(a_{i},a_{-i};b^{*}(a_{i},a_{-i}),m)))=0$,
while this KL divergence is strictly positive for any other choice
of $b$.

Now we show that Example \ref{exa:stability_reversal_example}  exhibits stability reversal. 
In every EZ with $p=(1,0),$ we must have $a_{AA}=a_{AB}=1.$
If $a_{BA}=2,$ then the adherents of $\Theta_{B}$ infer $b^{*}(1,2)=b^{\bullet}+\frac{m}{3}$.
With this inference, the biased agents expect $1\cdot(2(b^{\bullet}+\frac{m}{3})-m)=2b^{\bullet}-\frac{m}{3}$
from playing 1 against rival investment 1, and expect $2\cdot(3(b^{\bullet}+\frac{m}{3})-m)-c=6b^{\bullet}-c$
from playing 2 against rival investment 1. Since $4b^{\bullet}+\frac{m}{3}-c>0$
from Condition \ref{cond:large_misspec}, there is an EZ with $a_{BA}=2$
and $\mu_{B}$ puts probability 1 on $b^{\bullet}+\frac{m}{3}$. It
is impossible to have $a_{BA}=1$ in EZ. This is because $b^{*}(1,1)>b^{*}(1,2),$
and under the inference $b^{*}(1,2)$ we already have that the best
response to 1 is 2, so the same also holds under any higher belief
about complementarity. Also, we have $a_{BB}=2$, since 2 must best
respond to both 1 and 2. So in every such EZ, $\Theta_{A}$'s conditional
fitness against group A is $2b^{\bullet}$ and $\Theta_{B}$'s conditional
fitness against group A is $6b^{\bullet}-c$, with $2b^{\bullet}>6b^{\bullet}-c$
by Condition \ref{cond:medium_cost}. Also, $\Theta_{A}$'s conditional
fitness against group B is $3b^{\bullet}$, while $\Theta_{B}$'s
conditional fitness against group B is $8b^{\bullet}-c$. Again, $3b^{\bullet}>8b^{\bullet}-c$
by Condition \ref{cond:medium_cost}.

Next, we show $\Theta_{B}$ has strictly higher fitness than $\Theta_{A}$
in every EZ with $p_{B}=1.$ There is no EZ with $a_{BB}=1.$
This is because $b^{*}(1,1)=b^{\bullet}+\frac{m}{2}$. As discussed
before, under this inference the best response to 1 is 2, not 1. Now
suppose $a_{BB}=2.$ Then $\mu_{B}$ puts probability 1 on $b^{*}(2,2)=b^{\bullet}+\frac{m}{4}.$
With this inference, the biased agents expect $1\cdot(3(b^{\bullet}+\frac{m}{4})-m)=3b^{\bullet}-\frac{m}{4}$
from playing 1 against rival investment 2, and expect $2\cdot(4(b^{\bullet}+\frac{m}{4})-m)-c=8b^{\bullet}-c$
from playing 2 against rival investment 2. We have $5b^{\bullet}+\frac{m}{4}-c>0$
from Condition \ref{cond:large_misspec}, so 2 best responds to 2.
We must have $a_{AA}=a_{AB}=1.$ We conclude the unique EZ behavior
is $(a_{AA},a_{AB},a_{BA},a_{BB})=(1,1,1,2)$, since the biased agents
expect $1\cdot(2(b^{\bullet}+\frac{m}{4})-m)=2b^{\bullet}-\frac{m}{2}$
from playing 1 against rival investment 1, and expect $2\cdot(3(b^{\bullet}+\frac{m}{4})-m)-c=6b^{\bullet}-\frac{m}{2}-c$
from playing 2 against rival investment 1. We have $4b^{\bullet}-c<0$
from Condition \ref{cond:medium_cost}, so 1 best responds to 1. In
the unique EZ with $p=(0,1),$ the fitness of $\Theta_{A}$
is $2b^{\bullet}$ and the fitness of $\Theta_{B}$ is $8b^{\bullet}-c,$
where $8b^{\bullet}-c>2b^{\bullet}$ by Condition \ref{cond:medium_cost}.

\subsection{Proof of Proposition \ref{prop:abee}}
\begin{proof}
When $\Theta_{A}=\Theta_{B}=\Theta^{\bullet}$, with any $(p_{A},p_{B})$, we show adherents of both
models have 0 fitness in every  EZ. Suppose instead that the match
between groups $g$ and $g^{'}$ reach a terminal node other than
$z_{1}$ with positive probability. Let $n_{L}$ be the last non-terminal
node reached with positive probability, so we must have $L\ge2$,
and also that nodes $n_{1},...,n_{L-1}$ are also reached with positive
probability. So Drop must be played with probability 1 at $n_{L}.$
Since $n_{L}$ is reached with positive probability, correctly specified
agents hold correct beliefs about opponent's play at $n_{L}$, which
means at $n_{L-1}$ it cannot be optimal to play Across with positive
probability since this results in a loss of $\ell$ compared to playing
Drop, a contradiction.

Now let $\Theta_{A}=\Theta^{\bullet}$, $\Theta_{B}=\Theta^{An}$ and let $p_{B}\in(0,1).$ We claim there
is an EZ where $d_{AA}^{k}=1$ for every $k$, $d_{AB}^{k}=0$ for
every even $k$ with $k<K$, $d_{AB}^{k}=1$ for every other $k$,
$d_{BA}^{k}=0$ for every odd $k$ and $d_{BA}^{k}=1$ for every even
$k$, and $d_{BB}^{k}=0$ for every $k$ with $k<K,$ $d_{BB}^{K}=1.$
It is easy to see that the behavior $(d_{AA})$ is optimal under correct
belief about opponent's play. In the $\Theta_{A}$ vs. $\Theta_{B}$
matches, the conjecture about A's play $\hat{d}_{AB}^{k}=2/K$ for
$k$ even, $\hat{d}_{AB}^{k}=1$ for $k$ odd minimizes KL divergence
among all strategies in $\mathbb{A}^{An}$, given B's play. To see
this, note that when B has the role of P2, opponent Drops immediately.
When B has the role of P1, the outcome is always $z_{K}.$ So a conjecture
with $\hat{d}_{AB}^{k}=x$ for every even $k$ has the conditional
KL divergence of: 
\begin{align*}
 & \sum_{k\le K-1\text{ odd}}\underset{(1,z_{k})\text{ for }k\le K-1\text{ odd}}{\underbrace{0\cdot\ln\left(\frac{0}{0}\right)}}+\sum_{k\le K-1\text{ even}}\underset{(1,z_{k})\text{ for }k\le K-1\text{ even}}{\underbrace{0\cdot\ln\left(\frac{0}{(1/2)\cdot(1-x)^{(k/2)-1}\cdot x}\right)}}\\
 & +\underset{(1,z_{K})}{\underbrace{\frac{1}{2}\ln\left(\frac{1/2}{(1/2)\cdot(1-x)^{(K/2)-1}\cdot x}\right)}}+\underset{(1,z_{end})}{\underbrace{0\cdot\ln\left(\frac{0}{(1-x)^{(K/2)}}\right)}}
\end{align*}
when matched with an opponent from $\Theta_{A}$. Using $0\cdot\ln(0)=0,$
the expression simplifies to $\frac{1}{2}\ln\left(\frac{1}{(1-x)^{(K/2)-1}\cdot x}\right)$,
which is minimized among $x\in[0,1]$ by $x=2/K.$ Against this conjecture,
the difference in expected payoff at node $n_{K-1}$ from Across versus
Drop is $(1-2/K)(g)+(2/K)(-\ell).$ This is strictly positive when
$g>\frac{2}{K-2}\ell.$ This means the continuation value at $n_{K-1}$
is at least $g$ larger than the payoff of Dropping at $n_{K-3},$
so again Across has strictly higher expected payoff than Drop. Inductively,
$(d_{BA}^{k})$ is optimal given the belief $(\hat{d}_{AB}^{k}).$
Also, $(d_{AB}^{k})$ is optimal as it results in the highest possible
payoff. We can similarly show that the conjecture $\hat{d}_{BB}^{k}$
with $\hat{d}_{BB}^{k}=2/K$ for $k$ even, $\hat{d}_{BB}^{k}=0$
for $k$ odd minimizes KL divergence conditional on $\Theta_{B}$
opponent, and $(d_{BB}^{k})$ is optimal given this conjecture.

As $p_{B}\to0,$ we find an  EZ where adherents of A have fitness
approaching 0, whereas the adherents of B have fitness approaching at least $\frac{1}{2}(((K/2)-1)$$g-\ell)>0$
since $g>\frac{2}{K-2}\ell.$ This shows $\Theta_{A}$ is not evolutionarily
stable against $\Theta_{B}$.

But consider the same $(d_{AA},d_{AB},d_{BA})$ and suppose $d_{BB}^{k}=1$
for every $k$. Taking $p_{B}\to1,$ we find an
 EZ where adherents of B have fitness 0, adherents of A have fitness
$\frac{1}{2}\cdot((K/2)g+\ell)>0.$ This shows $\Theta_{B}$
is not evolutionarily stable against $\Theta_{A}$.
\end{proof}

\subsection{Proof of Proposition \ref{prop:stable_pop_share}}
\begin{proof}
Take $g>\frac{2}{K-2}\ell$ in the centipede game. The misspecified
agent thinks a group B agent in the role of P2 and a group A agent
in either role has a probability $2/K$ of stopping at every node.
Under this belief, choosing to continue instead of drop means there
is a $(K-2)/K$ chance of gaining $g$, but a $2/K$ chance of losing
$\ell$. Since we assume $g>\frac{2}{K-2}\ell$, it is strictly better
to continue. When $p$ fraction of the agents are correctly specified,
the fitness of $\Theta^{\bullet}$ is $p\cdot0+(1-p)\cdot(\frac{1}{2}\frac{g(K-2)}{2}+\frac{1}{2}(\frac{gK}{2}+\ell))$,
while the fitness of $\Theta^{An}$ is $p\cdot[\frac{1}{2}(\frac{g(K-2)}{2}-\ell)+\frac{1}{2}\frac{g(K-2)}{2}]+(1-p)[\frac{1}{2}(\frac{g(K-2)}{2}-\ell)+\frac{1}{2}(\frac{gK}{2}+\ell)]$.
The difference in fitness is $-p[\frac{1}{2}(\frac{g(K-2)}{2}-\ell)+\frac{1}{2}\frac{g(K-2)}{2}]+(1-p)\frac{1}{2}\ell.$
Simplifying, this is $\frac{1}{2}\ell-p\cdot\frac{g(K-2)}{2}$, a
strictly decreasing function in $p.$ When $p=\frac{\ell}{g(K-2)},$
which is a number strictly between 0 and 1/2 from the assumption $g>\frac{2}{K-2}\ell$
in the centipede game, the two models have the same fitness. Furthermore, since the payoff difference is linear in $p$ with a negative slope, the difference in fitness is negative when $p >\frac{\ell}{g(K-2)}$---so that $\Theta^{An}$ outperforms $\Theta^{\bullet}$ under these population shares---and conversely, the difference in fitness is positive when $p < \frac{\ell}{g(K-2)}$. Thus, we have this fraction of the population being correctly specified forms a stable population share.
\end{proof}

\subsection{Proof of Proposition \ref{prop:stable_pop_share_dollar}}
\begin{proof}
In the $\Theta^{An}$ vs. $\Theta^{An}$ match,
the adherents of $\Theta^{An}$ hold the belief that $\hat{d}_{BB}^{k}=2/K$
for every even $k$. In the role of P1, at node $k$ for $k\le K-3,$
stopping gives them $k$ but continuing gives them a $(K-2)/K$ chance
to get at least $k+2$, and we have $k\le\frac{K-2}{K}(k+2)\iff2k\le2K-4\iff k\le K-2$.
At node $K-1,$ the agent gets $K-1$ from dropping but expects $(K+2)\cdot\frac{K-2}{K}$
from continuing, and $(K+2)\cdot\frac{K-2}{K}-(K-1)=\frac{K^{2}-4-K^{2}+K}{K}=\frac{K-4}{K}>0$
since $K\ge6.$

In the $\Theta^{\bullet}$ vs. $\Theta^{An}$
match, the adherents of $\Theta^{An}$ hold the belief that $\hat{d}_{AB}^{k}=2/K$
for every $k.$ By the same arguments as before, the behavior of the
adherents of $\Theta^{An}$ are optimal given these beliefs. Also,
the adherents of $\Theta^{\bullet}$ have no profitable deviations
since they are best responding both as P1 and P2.

When $p$ fraction of the agents are correctly specified, in the dollar
game the fitness of $\Theta^{\bullet}$ is $p\cdot0.5+(1-p)\cdot(\frac{1}{2}(K-1)+\frac{1}{2}K)$,
while the fitness of $\Theta^{An}$ is $p\cdot0+(1-p)\cdot(\frac{1}{2}\cdot0+\frac{1}{2}K)$.
For any $p$, the fitness of $\Theta^{\bullet}$ is strictly
higher than that of $\Theta^{An}$.
\end{proof}

\section{Existence and Continuity of Equilibrium Zeitgeists}\label{sec:Existence-and-Continuity}

We provide a few technical results about the existence of EZ and the
upper-hemicontinuity of the set of EZs with respect to population
share. We suppose that $|\mathcal{G}|=1$ for simplicity, but analogous
results would hold for environments with multiple situations. Note
that the same belief endogeneity that generates new stability phenomena
in Section \ref{sec:new_stability_phenomena} also leads to some difficulty
in establishing existence and continuity results, as agents draw different
inferences with different rates of interactions with the various groups. 

We provide two sets of results. The first concerns environments where
the expected KL divergence of any parameter in the model is finite
under any strategy profile (for example, when every parameter conjectures
a full-support distribution over consequences in $\mathbb{Y}$ under
every strategy profile, and the support of the monitoring signal does
not vary with opponent's strategy). The second focuses on environments
with strategic certainty, so monitoring signals do not have full support
and instead perfectly reveal opponent's strategy. But, we impose the
same finite KL divergence requirement on the consequences. 

For each $g,g^{'}\in\{A,B\},$ define $K_{g,g^{'}}:\mathbb{A}^{2}\times\mathcal{G}\times\Theta_{g}\to\mathbb{R}$
by $K_{g,g^{'}}(a_{i},a_{-i},G;(a_{A},a_{B},F))=D_{KL}(F^{\bullet}(a_{i},a_{-i},G)\times\varphi^{\bullet}(a_{-i})\parallel F(a_{i},a_{g^{'}})\times\varphi^{\bullet}(a_{g'})).$
This is the KL divergence of the parameter $(a_{A},a_{B},F)\in\Theta_{g}$
in situation $G$ based on the data generated from the strategy profile
$(a_{i},a_{-i})$.

\subsection{Environments with Full-Support Monitoring Signals}\label{subsec:Environments-with-Full-Support}

Let two models, $\Theta_{A},\Theta_{B}$ be fixed. Also fix population
shares $p$$.$ For $g,g'\in\{A,B\},$ define $V_{g,g'}:\mathbb{A}^{2}\times\Theta_{g}\to\mathbb{R}$
to be $V_{g,g'}(a_{i},a_{-i},(\hat{a}_{A},\hat{a}_{B},F)):=\mathbb{E}_{y\sim F(a_{i},\hat{a}_{g'})}(\pi(y))$,
the expected payoff from choosing strategy $a_{i}$ under the parameter
$(\hat{a}_{A},\hat{a}_{B},F)$ when matched with an opponent from
group $g'$. Extend the domain of the third argument of $V_{g,g'}$
from $\Theta_{g}$ to $\Delta(\Theta_{g})$ by linearity. The $K_{g,g'}$ function defined above
specializes in the case   of $|\mathcal{G}|=1$ to be {$K_{g,g^{'}}(a_{i},a_{-i};(\hat{a}_{A},\hat{a}_{B},F))=D_{KL}(F^{\bullet}(a_{i},a_{-i})\times\varphi^{\bullet}(a_{-i})\parallel F(a_{i},\hat{a}_{g^{'}})\times\varphi^{\bullet}(\hat{a}_{g'})).$ }
\begin{assumption}
\label{assu:compact}$\mathbb{A},\Theta_{A},\Theta_{B}$ are compact
metrizable spaces.
\end{assumption}
\begin{assumption}
\label{assu:continuous_utility} For every $g,g'\in\{A,B\},$ $V_{g,g'}$
is continuous.
\end{assumption}
\begin{assumption}
\label{assu:finite_KL}For every $g,g'\in\{A,B\},$ $K_{g,g^{'}}$
is well-defined and finite on its domain $\mathbb{A}^{2}\times\Theta_{g}$. 
\end{assumption}
\selectlanguage{american}%
\selectlanguage{english}%
\begin{assumption}
\label{assu:continuous_KL} For every $g,g'\in\{A,B\},$ $K_{g,g^{'}}$
is continuous.
\end{assumption}
\begin{assumption}
\label{assu:quasiconcave}$\mathbb{A}$ is convex and, for $g,g'\in\{A,B\}$,
all $a_{-i}\in\mathbb{A}$ and all $\mu_{g}\in\Delta(\Theta_{g}),$
$a_{i}\mapsto V_{g,g'}(a_{i},a_{-i};\mu_{g})$ is quasiconcave. 
\end{assumption}
We show the existence of equilibrium zeitgeists using the Kakutani-Fan-Glicksberg
fixed point theorem, applied to the correspondence which maps strategy
profiles and beliefs over parameters into best replies and beliefs
over KL-divergence minimizing parameter. We start with a lemma.
\begin{lem}
\label{lem:inference_uhc}For $g\in\{A,B\}$, $a=(a_{AA},a_{AB},a_{BA},a_{BB})\in\mathbb{A}^{4},$
and $0\le p_{g}\le1$, let 
\[
\Theta_{g}^{*}(a,p_{g}):=\underset{\hat{\theta}\in\Theta_{g}}{\arg\min}\left\{ \begin{array}{c}
p_{g}\cdot K_{g,g}(a_{g,g},a_{g,g};\hat{\theta})+(1-p_{g})\cdot K_{g,-g}(a_{g,-g},a_{-g,g};\hat{\theta})\end{array}\right\} .
\]
Then, $\Theta_{g}^{*}$ is upper hemicontinuous in its arguments. 
\end{lem}
This lemma says the set of KL-minimizing parameters is upper hemicontinuous
in strategy profile and population share. This leads to the existence
result. 
\begin{prop}
\label{prop:existence}Under Assumptions \ref{assu:compact}, \ref{assu:continuous_utility},
\ref{assu:finite_KL}, \ref{assu:continuous_KL}, and \ref{assu:quasiconcave},
an equilibrium zeitgeist exists. 
\end{prop}
Next, upper hemicontinuity in $p_{g}$ in Lemma \ref{lem:inference_uhc}
allows us to deduce the upper hemicontinuity of the EZ correspondence
in population shares. 
\begin{prop}
\label{prop:uhc}Fix two models $\Theta_{A},\Theta_{B}$. The set
of equilibrium zeitgeists is an upper hemicontinuous correspondence
in $p_{B}$ under Assumptions \ref{assu:compact}, \ref{assu:continuous_utility},
\ref{assu:finite_KL}, and \ref{assu:continuous_KL}. 
\end{prop}

\subsection{Proofs of Results in Appendix \ref{subsec:Environments-with-Full-Support}}

\subsubsection{Proof of Lemma \ref{lem:inference_uhc}}
\begin{proof}
Write the minimization objective as 
\[
W(a,p_{g},\hat{\theta}):=p_{g}\cdot K_{g,g}(a_{g,g},a_{g,g};\hat{\theta})+(1-p_{g})\cdot K_{g,-g}(a_{g,-g},a_{-g,g};\hat{\theta}),
\]
a continuous function of $(a,p_{g},\hat{\theta})$ by Assumption \ref{assu:continuous_KL}.
Suppose we have a sequence $(a^{(n)},p_{g}^{(n)})\to(a^{*},p_{g}^{*})\in\mathbb{A}^{4}\times[0,1]$
and let $\hat{\theta}^{(n)}\in\Theta_{g}^{*}(a^{(n)},p_{g}^{(n)})$
for each $n,$ with $\theta^{(n)}\to\theta^{*}\in\Theta_{g}.$ For
any other $\theta'\in\Theta_{g},$ note that $W(a^{*},p_{g}^{*},\theta')=\lim_{n\to\infty}W(a^{(n)},p_{g}^{(n)},\theta')$
by continuity. But also by continuity, $W(a^{*},p_{g}^{*},\theta^{*})=\lim_{n\to\infty}W(a^{(n)},p_{g}^{(n)},\hat{\theta}^{(n)})$
and $W(a^{(n)},p_{g}^{(n)},\hat{\theta}^{(n)})\le W(a^{(n)},p_{g}^{(n)},\theta')$
for every $n.$ It therefore follows $W(a^{*},p_{g}^{*},\theta^{*})\le W(a^{*},p_{g}^{*},\theta').$ 
\end{proof}

\subsubsection{Proof of Proposition \ref{prop:existence}}
\begin{proof}
Consider the correspondence $\Gamma:\mathbb{A}^{4}\times\Delta(\Theta_{A})\times\Delta(\Theta_{B})\rightrightarrows\mathbb{A}^{4}\times\Delta(\Theta_{A})\times\Delta(\Theta_{B}),$
\begin{align*}
 & \Gamma(a_{AA},a_{AB},a_{BA},a_{BB},\mu_{A},\mu_{B}):=\\
 & (\text{BR}_{AA}(\mu_{A}),\text{BR}_{AB}(\mu_{A}),\text{BR}_{BA}(\mu_{B}),\text{BR}_{BB}(\mu_{B}),\Delta(\Theta_{A}^{*}(a)),\Delta(\Theta_{B}^{*}(a))),
\end{align*}
where $\text{BR}_{gg'}(\mu_{g}):=\underset{a_{i}\in\mathbb{A}}{\arg\max}V_{g,g'}(a_{i},a_{-i};\mu_{g})$
(this is well-defined because $V_{g,g'}$ does not depend on its second
argument) and, for each $g\in\{A,B\},$ we have omitted the dependence
of the correspondence $\Theta_{g}^{*}$ on $p_{g}$. It is clear that
fixed points of $\Gamma$ are equilibrium zeitgeists.

We apply the Kakutani-Fan-Glicksberg theorem (see, e.g, Corollary
17.55 in \cite{aliprantis_infinite_2006}). By Assumptions \ref{assu:compact}
and \ref{assu:quasiconcave}, $\mathbb{A}$ is a compact and convex
metric space, and each $\Theta_{g}$ is a compact metric space, so
it follows the domain of $\Gamma$ is a nonempty, compact and convex
metric space. We need only verify that $\Gamma$ has closed graph,
non-empty values, and convex values.

To see that $\Gamma$ has closed graph, the previous lemma shows the
upper hemicontinuity of $\Theta_{A}^{*}(a)$ and $\Theta_{B}^{*}(a)$
in $a,$ and Theorem 17.13 of \cite{aliprantis_infinite_2006} then
implies $\Delta(\Theta_{A}^{*}(a))$ and $\Delta(\Theta_{B}^{*}(a))$
are also upper hemicontinuous in $a.$ It is a standard argument that
since Assumption \ref{assu:continuous_utility} supposes $V_{AA},V_{AB},V_{BA,}V_{BB}$
are continuous, it implies the best-response correspondences $\text{BR}_{AA}(\mu_{A}),$
$\text{BR}_{AB}(\mu_{A}),$ $\text{BR}_{BA}(\mu_{B}),$ $\text{BR}_{BB}(\mu_{B})$
have closed graphs.

To see that $\Gamma$ is non-empty, recall that each $a_{i}\mapsto V_{g,g'}(a_{i},a_{-i};\mu_{g})$
is a continuous function on a compact domain, so it must attain a
maximum on $\mathbb{A}.$ Similarly, the minimization problem that
defines each $\Theta_{g}^{*}(a)$ is a continuous function of the
parameter over a compact domain of possible parameters, so it attains
a minimum. Thus each $\Delta(\Theta_{g}^{*}(a))$ is the set of distributions
over a non-empty set.

To see that $\Gamma$ is convex valued, clearly $\Delta(\Theta_{A}^{*}(a))$
and $\Delta(\Theta_{B}^{*}(a))$ are convex valued by definition.
Also, $a_{i}\mapsto V_{AA}(a_{i},a_{-i};\mu_{A})$ is quasiconcave
by Assumption \ref{assu:quasiconcave}. That means if $a_{i}^{'},a_{i}^{''}\in\text{BR}_{AA}(\mu_{A}),$
then for any convex combination $\tilde{a}_{i}$ of $a_{i}^{'},a_{i}^{''},$
we have{\small{} $V_{AA}(\tilde{a}_{i},a_{-i};\mu_{A})\ge\min(V_{AA}(a_{i}^{'},a_{-i};\mu_{A}),$}
$V_{AA}(a_{i}^{''},a_{-i};\mu_{A}))=\max_{a_{i}\in\mathbb{A}}V_{AA}(a_{i},a_{-i};\mu_{A})$.
Therefore, $\text{BR}_{AA}(\mu_{A})$ is convex. For similar reasons,
$\text{BR}_{AB}(\mu_{A}),$ $\text{BR}_{BA}(\mu_{B}),$ $\text{BR}_{BB}(\mu_{B})$
are convex. 
\end{proof}

\subsubsection{Proof of Proposition \ref{prop:uhc}}
\begin{proof}
Since $\mathbb{A}^{4}\times\Delta(\Theta_{A})\times\Delta(\Theta_{B})$
is compact by Assumption \ref{assu:compact}, we need only show that
for every sequence $(p_{B}^{(k)})_{k\ge1}$ and $(a^{(k)},\mu^{(k)})_{k\ge1}=(a_{AA}^{(k)},a_{AB}^{(k)},a_{BA}^{(k)},a_{BB}^{(k)},\mu_{A}^{(k)},\mu_{B}^{(k)})_{k\ge1}$
such that for every $k$, $(a^{(k)},\mu^{(k)})$ is an EZ with $p=(1-p_{B}^{(k)},p_{B}^{(k)})$,
$p_{B}^{(k)}\to p_{B}^{*}$, and $(a^{(k)},\mu^{(k)})\to(a^{*},\mu^{*})$,
then $(a^{*},\mu^{*})$ is an EZ with $p=(1-p_{B}^{*},p_{B}^{*})$.

We first show for all $g,g^{'}\in\{A,B\},$ $a_{g,g^{'}}^{*}$ is
optimal under the belief $\mu_{g}^{*}.$ By Assumption \ref{assu:continuous_utility},
$V_{g,g'}(a_{i},a_{-i};\mu_{g})$ is continuous, so by the property of
convergence in distribution, $V_{g,g'}(a_{g,g^{'}}^{(k)},a_{g^{'},g}^{(k)};\mu_{g}^{(k)})\to V_{g,g'}(a_{g,g^{'}}^{*},a_{g^{'},g}^{*};\mu_{g}^{*})$.
For any other $a_{i}'\in\mathbb{A},$ $V_{g,g'}(a_{i}',a_{g^{'},g}^{(k)};\mu_{g}^{(k)})\to V_{g,g'}(a_{i}',a_{g^{'},g}^{*};\mu_{g}^{*})$
and for every $k,$ $V_{g,g'}(a_{g,g^{'}}^{(k)},a_{g^{'},g}^{(k)};\mu_{g}^{(k)})\ge V_{g,g'}(a_{i}',a_{g^{'},g}^{(k)};\mu_{g}^{(k)}).$
Therefore $a_{g,g^{'}}^{*}$ best responds to the belief $\mu_{g}^{*}.$

Next, we show parameters in the support of $\mu_{g}^{*}$ minimize
weighted KL divergence for group $g.$ Since $\Theta_{g}^{*}(a,p_{g})$
represents the minimizers of a continuous function on a compact domain
(by Lemma \ref{lem:inference_uhc}), it is non-empty and closed. By
Theorem 17.13 of \cite{aliprantis_infinite_2006}, the correspondence
$\tilde{H}:\mathbb{A}^{4}\times[0,1]\rightrightarrows\Delta(\Theta_{g})$
defined so that $\tilde{H}(a,p_{g}):=\Delta(\Theta_{g}^{*}(a,p_{g}))$
is also upper hemicontinuous. For every $k,$ $\mu_{g}^{(k)}\in\tilde{H}(a^{(k)},p_{g}^{(k)})$,
and $\mu_{g}^{(k)}\to\mu_{g}^{*}$, $a^{(k)}\to a^{*},p_{g}^{(k)}\to p_{g}^{*}.$
Therefore, $\mu_{g}^{*}\in\tilde{H}(a^{*},p_{g}^{*}),$ that is to
say $\mu_{g}^{*}$ is supported on the minimizers of weighted KL divergence. 
\end{proof}

\subsection{Environments with Strategic Certainty}\label{subsec:Environments-with-Strategic-Cer}

Let two models, $\Theta_{A},\Theta_{B}$ be fixed. Suppose we are
in an environment with strategic certainty, so each $\Theta_{g}$
has the form $\mathbb{A}^{2}\times\mathcal{F}_{g}$, $\mathbb{M}=\mathbb{A},$
and for every $a_{-i}\in\mathbb{A}$, $\varphi^{\bullet}(a_{-i})$
puts probability 1 on $a_{-i}$. Fix population shares $p$$.$ As
discussed in Section \ref{subsec:EZDef}, \foreignlanguage{american}{we
omit the part of the parameter that corresponds to conjectures about
others' strategies and simply view beliefs as elements in $\Delta(\mathcal{F}_{g}).$ }

For $g\in\{A,B\},$ let $U_{g}:\mathbb{A}^{2}\times\mathcal{F}_{g}\to\mathbb{R}$
be defined by $U_{g}(a_{i},a_{-i};F)=U_{i}(a_{i},a_{-i};\delta_{F})$
with $U_{i}:\mathbb{A}^{2}\times(\Delta(\mathcal{F}_{A})\cup\Delta(\mathcal{F}_{B}))$
as defined before. Extend the domain of the third argument of $U_{g}$
from $\mathcal{F}_{g}$ to $\Delta(\mathcal{F}_{g})$ by linearity.
Also,  the $K(a_{i},a_{-i};F)$ function specializes in the case of  $|\mathcal{G}|=1$ to be {$K(a_{i},a_{-i};F)=D_{KL}(F^{\bullet}(a_{i},a_{-i})\parallel F(a_{i},a_{-i})).$ }
\begin{assumption}
\label{assu:compact-sc}$\mathbb{A},\mathcal{F}_{A},\mathcal{F}_{B}$
are compact metrizable spaces.
\end{assumption}
\begin{assumption}
\label{assu:continuous_utility-sc} For each $g\in\{A,B\}$$,$ $U_{g}$
is continuous.
\end{assumption}
\begin{assumption}
\label{assu:finite_KL-sc} $K$ is well-defined and finite on its
domain $\mathbb{A}^{2}\times(\mathcal{F}_{A}\cup\mathcal{F}_{B})$. 
\end{assumption}
Under Assumption \ref{assu:finite_KL-sc}, we can define $K_{A}:\mathbb{A}^{2}\times\mathcal{F}_{A}\to\mathbb{R}$
and $K_{B}:\mathbb{A}^{2}\times\mathcal{F}_{B}\to\mathbb{R}$, with
$K_{g}(a_{i},a_{-i};F)=K(a_{i},a_{-i};F)$ for each $g\in\{A,B\},$
$a_{i},a_{-i}\in\mathbb{A},$ and $F\in\mathcal{F}_{g}.$ 

\selectlanguage{american}%
\selectlanguage{english}%
\begin{assumption}
\label{assu:continuous_KL-sc} For each $g\in\{A,B\},$ $K_{g}$ is
continuous.
\end{assumption}
\begin{assumption}
\label{assu:quasiconcave-sc}$\mathbb{A}$ is convex and, for $g\in\{A,B\}$,
all $a_{-i}\in\mathbb{A}$ and all $\mu_{g}\in\Delta(\mathcal{F}_{g}),$
$a_{i}\mapsto U_{g}(a_{i},a_{-i};\mu_{g})$ is quasiconcave. 
\end{assumption}
We show the existence of equilibrium zeitgeists using the Kakutani-Fan-Glicksberg
fixed point theorem, applied to the correspondence which maps strategy
profiles and beliefs over parameters into best replies and beliefs
over KL-divergence minimizing parameter. We start with a lemma.
\begin{lem}
\label{lem:inference_uhc-sc}For $g\in\{A,B\}$, $a=(a_{AA},a_{AB},a_{BA},a_{BB})\in\mathbb{A}^{4},$
and $0\le p_{g}\le1$, let 
\[
\Theta_{g}^{*}(a,p_{g}):=\underset{\hat{F}\in\mathcal{F}_{g}}{\arg\min}\left\{ \begin{array}{c}
p_{g}\cdot K_{g}(a_{g,g},a_{g,g};\hat{F})+(1-p_{g})\cdot K_{g,}(a_{g,-g},a_{-g,g};\hat{F})\end{array}\right\} .
\]
Then, $\Theta_{g}^{*}$ is upper hemicontinuous in its arguments. 
\end{lem}
This lemma says the set of KL-minimizing parameters is upper hemicontinuous
in strategy profile and population share. This leads to the existence
result. 
\begin{prop}
\label{prop:existence-sc}Under Assumptions \ref{assu:compact-sc},
\ref{assu:continuous_utility-sc}, \ref{assu:finite_KL-sc}, \ref{assu:continuous_KL-sc},
and \ref{assu:quasiconcave-sc}, an equilibrium zeitgeist exists. 
\end{prop}
Next, upper hemicontinuity in $p_{g}$ in Lemma \ref{lem:inference_uhc-sc}
allows us to deduce the upper hemicontinuity of the EZ correspondence
in population shares. 
\begin{prop}
\label{prop:uhc-sc}Fix two models $\Theta_{A},\Theta_{B}$. The set
of equilibrium zeitgeists is an upper hemicontinuous correspondence
in $p_{B}$ under Assumptions \ref{assu:compact-sc}, \ref{assu:continuous_utility-sc},
\ref{assu:finite_KL-sc}, and \ref{assu:continuous_KL-sc}. 
\end{prop}

\subsection{Proofs of Results in Appendix \ref{subsec:Environments-with-Strategic-Cer}}

\subsubsection{Proof of Lemma \ref{lem:inference_uhc-sc}}
\begin{proof}
Write the minimization objective as 
\[
W(a,p_{g},\hat{F}):=p_{g}\cdot K_{g}(a_{g,g},a_{g,g};\hat{F})+(1-p_{g})\cdot K_{g,}(a_{g,-g},a_{-g,g};\hat{F}),
\]
a continuous function of $(a,p_{g},\hat{F})$ by Assumption \ref{assu:continuous_KL-sc}.
Suppose we have a sequence $(a^{(n)},p_{g}^{(n)})\to(a^{*},p_{g}^{*})\in\mathbb{A}^{4}\times[0,1]$
and let $F^{(n)}\in\Theta_{g}^{*}(a^{(n)},p_{g}^{(n)})$ for each
$n,$ with $F^{(n)}\to F^{*}\in\mathcal{F}_{g}.$ For any other $F'\in\mathcal{F}_{g},$
note that $W(a^{*},p_{g}^{*},F')=\lim_{n\to\infty}W(a^{(n)},p_{g}^{(n)},F')$
by continuity. But also by continuity, $W(a^{*},p_{g}^{*},F^{*})=\lim_{n\to\infty}W(a^{(n)},p_{g}^{(n)},F^{(n)})$
and $W(a^{(n)},p_{g}^{(n)},F^{(n)})\le W(a^{(n)},p_{g}^{(n)},F')$
for every $n.$ It therefore follows $W(a^{*},p_{g}^{*},F^{*})\le W(a^{*},p_{g}^{*},F').$ 
\end{proof}

\subsubsection{Proof of Proposition \ref{prop:existence-sc}}
\begin{proof}
Consider the correspondence $\Gamma:\mathbb{A}^{4}\times\Delta(\mathcal{F}_{A})\times\Delta(\mathcal{F}_{B})\rightrightarrows\mathbb{A}^{4}\times\Delta(\mathcal{F}_{A})\times\Delta(\mathcal{F}_{B}),$
\begin{align*}
 & \Gamma(a_{AA},a_{AB},a_{BA},a_{BB},\mu_{A},\mu_{B}):=\\
 & (\text{BR}(a_{AA},\mu_{A}),\text{BR}(a_{BA},\mu_{A}),\text{BR}(a_{AB},\mu_{B}),\text{BR}(a_{BB},\mu_{B}),\Delta(\Theta_{A}^{*}(a)),\Delta(\Theta_{B}^{*}(a))),
\end{align*}
where $\text{BR}(a_{-i},\mu_{g}):=\underset{\hat{a}_{i}\in\mathbb{A}}{\arg\max}U_{g}(\hat{a}_{i},a_{-i};\mu_{g})$
and, for each $g\in\{A,B\},$ we have omitted the dependence of the
correspondence $\Theta_{g}^{*}$ on $p_{g}$. It is clear that fixed
points of $\Gamma$ are equilibrium zeitgeists.

We apply the Kakutani-Fan-Glicksberg theorem (see, e.g., Corollary
17.55 in \cite{aliprantis_infinite_2006}). By Assumptions \ref{assu:compact-sc}
and \ref{assu:quasiconcave-sc}, $\mathbb{A}$ is a compact and convex
metric space, and each $\mathcal{F}_{g}$ is a compact metric space,
so it follows the domain of $\Gamma$ is a nonempty, compact and convex
metric space. We need only verify that $\Gamma$ has closed graph,
non-empty values, and convex values.

To see that $\Gamma$ has closed graph, the previous lemma shows the
upper hemicontinuity of $\Theta_{A}^{*}(a)$ and $\Theta_{B}^{*}(a)$
in $a,$ and Theorem 17.13 of \cite{aliprantis_infinite_2006} then
implies $\Delta(\Theta_{A}^{*}(a))$ and $\Delta(\Theta_{B}^{*}(a))$
are also upper hemicontinuous in $a.$ It is a standard argument that
since Assumption \ref{assu:continuous_utility-sc} supposes $U_{A},U_{B}$
are continuous, it implies the best-response correspondence $\text{BR}$
has closed graph.

To see that $\Gamma$ is non-empty, recall that each $a_{i}\mapsto U_{g}(a_{i},a_{-i};\mu_{g})$
is a continuous function on a compact domain, so it must attain a
maximum on $\mathbb{A}.$ Similarly, the minimization problem that
defines each $\Theta_{g}^{*}(a)$ is a continuous function of the
parameter over a compact domain of possible parameters, so it attains
a minimum. Thus each $\Delta(\Theta_{g}^{*}(a))$ is the set of distributions
over a non-empty set.

To see that $\Gamma$ is convex valued, clearly $\Delta(\Theta_{A}^{*}(a))$
and $\Delta(\Theta_{B}^{*}(a))$ are convex valued by definition.
Also, $a_{i}\mapsto U_{g}(a_{i},a_{-i};\mu_{g})$ is quasiconcave
by Assumption \ref{assu:quasiconcave-sc}. That means if $a_{i}^{'},a_{i}^{''}\in\text{BR}(a_{-i},\mu_{g}),$
then for any convex combination $\tilde{a}_{i}$ of $a_{i}^{'},a_{i}^{''},$
we have{\small{} $U_{g}(\tilde{a}_{i},a_{-i};\mu_{g})\ge\min(U_{g}(a_{i}^{'},a_{-i};\mu_{g}),$}
$U_{g}(a_{i}^{''},a_{-i};\mu_{g}))=\max_{a_{i}\in\mathbb{A}}U_{g}(a_{i},a_{-i};\mu_{g})$.
Therefore, $\text{BR}(a_{-i},\mu_{g})$ is convex. 
\end{proof}

\subsubsection{Proof of Proposition \ref{prop:uhc-sc}}
\begin{proof}
Since $\mathbb{A}^{4}\times\Delta(\mathcal{F}_{A})\times\Delta(\mathcal{F}_{B})$
is compact by Assumption \ref{assu:compact-sc}, we need only show
that for every sequence $(p_{B}^{(k)})_{k\ge1}$ and $(a^{(k)},\mu^{(k)})_{k\ge1}=(a_{AA}^{(k)},a_{AB}^{(k)},a_{BA}^{(k)},a_{BB}^{(k)},\mu_{A}^{(k)},\mu_{B}^{(k)})_{k\ge1}$
such that for every $k$, $(a^{(k)},\mu^{(k)})$ is an EZ with $p=(1-p_{B}^{(k)},p_{B}^{(k)})$,
$p_{B}^{(k)}\to p_{B}^{*}$, and $(a^{(k)},\mu^{(k)})\to(a^{*},\mu^{*})$,
then $(a^{*},\mu^{*})$ is an EZ with $p=(1-p_{B}^{*},p_{B}^{*})$.

We first show for all $g,g^{'}\in\{A,B\},$ $a_{g,g^{'}}^{*}$ is
optimal under the belief $\mu_{g}^{*}.$ By Assumption \ref{assu:continuous_utility-sc},
$U_{g}(a_{i},a_{-i};\mu_{g})$ is continuous, so by property of convergence
in distribution, $U_{g}(a_{g,g^{'}}^{(k)},a_{g^{'},g}^{(k)};\mu_{g}^{(k)})\to U_{g}(a_{g,g^{'}}^{*},a_{g^{'},g}^{*};\mu_{g}^{*})$.
For any other $a_{i}'\in\mathbb{A},$ $U_{g}(a_{i}',a_{g^{'},g}^{(k)};\mu_{g}^{(k)})\to U_{g}(a_{i}',a_{g^{'},g}^{*};\mu_{g}^{*})$
and for every $k,$ $U_{g}(a_{g,g^{'}}^{(k)},a_{g^{'},g}^{(k)};\mu_{g}^{(k)})\ge U_{g}(a_{i}',a_{g^{'},g}^{(k)};\mu_{g}^{(k)}).$
Therefore $a_{g,g^{'}}^{*}$ best responds to the belief $\mu_{g}^{*}.$

Next, we show parameters in the support of $\mu_{g}^{*}$ minimize
weighted KL divergence for group $g.$ Since $\Theta_{g}^{*}(a,p_{g})$
represents the minimizers of a continuous function on a compact domain
(by Lemma \ref{lem:inference_uhc-sc}), it is non-empty and closed.
By Theorem 17.13 of \cite{aliprantis_infinite_2006}, the correspondence
$\tilde{H}:\mathbb{A}^{4}\times[0,1]\rightrightarrows\Delta(\mathcal{F}_{g})$
defined so that $\tilde{H}(a,p_{g}):=\Delta(\Theta_{g}^{*}(a,p_{g}))$
is also upper hemicontinuous. For every $k,$ $\mu_{g}^{(k)}\in\tilde{H}(a^{(k)},p_{g}^{(k)})$,
and $\mu_{g}^{(k)}\to\mu_{g}^{*}$, $a^{(k)}\to a^{*},p_{g}^{(k)}\to p_{g}^{*}.$
Therefore, $\mu_{g}^{*}\in\tilde{H}(a^{*},p_{g}^{*}),$ that is to
say $\mu_{g}^{*}$ is supported on the minimizers of weighted KL divergence. \selectlanguage{american}%
\end{proof}

\section{Learning Foundation of Equilibrium Zeitgeists}\label{sec:Learning-Foundation}

We provide a foundation for equilibrium zeitgeists as the steady state
of a learning system. This foundation considers a world where agents
start with prior beliefs over parameters in a model. As in our framework
from Section \ref{sec:Environment-and-Stability}, these parameters
correspond to conjectures about the stage game and about others' strategies. 

At the end of every match, each agent observes their consequence and
a monitoring signal. We show that under any asymptotically myopic
policy, if behavior and beliefs converge, then the limit steady state
must be an EZ. If the models allow agents to make rich enough inferences
about opponents' strategies, then sufficiently accurate monitoring
signals about opponent's play imply that agents must hold correct
beliefs about others' strategies in the limit steady state. In particular,
in environments that approach strategic certainty (that is, the monitoring
signals are full support but they almost perfectly reveal opponent's
strategy), limit steady state beliefs about others' strategies must
be correct. Finally, if the true situation is redrawn every $T$ periods
and the agents reset their beliefs over parameters to their prior
belief when the situation is redrawn, then their average payoffs approach
their fitness in the EZ when $T$ is large.

\subsection{Regularity Assumptions}\label{subsec:Regularity-Assumptions-learning}

We make some regularity assumptions on the objective environments
and on the two models $\Theta_{A},\Theta_{B}$. These are similar
to the regularity assumptions from Appendix \ref{sec:Existence-and-Continuity}.

Suppose the strategy set $\mathbb{A}$ and the space of monitoring
signals $\mathbb{M}$ are finite. Suppose the marginals of the 
models $\Theta_{A},\Theta_{B}$ on the dimension of fundamental uncertainty,
denoted as $\mathcal{F}_{A},\mathcal{F}_{B}$, are compact and metrizable
spaces. Endow $\Theta_{A}$ and $\Theta_{B}$ with the product metric.
Suppose that every $(a_{A},a_{B},F)\in\Theta_{A}\cup\Theta_{B}$ is such that for every $(a_{i},a_{-i})\in\mathbb{A}^{2}$ and every situation
$G,$ whenever $f^{\bullet}(a_{i},a_{-i},G)(y)>0$, we also get $f(a_{i},a_{A})(y)>0$
and $f(a_{i},a_{B})(y)>0$, where $f$ is the density or probability
mass function for $F$. Suppose the monitoring signal has full support
on $\mathbb{M}$ for every $a_{-i}\in\mathbb{A}.$ 

For each $g,g^{'}\in\{A,B\},$ recall that we defined $K_{g,g^{'}}:\mathbb{A}^{2}\times\mathcal{G}\times\Theta_{g}\to\mathbb{R}$
in Appendix \ref{sec:Existence-and-Continuity}  by $K_{g,g^{'}}(a_{i},a_{-i},G;(a_{A},a_{B},F))=D_{KL}(F^{\bullet}(a_{i},a_{-i},G)\times\varphi^{\bullet}(a_{-i})\parallel F(a_{i},a_{g^{'}})\times\varphi^{\bullet}(a_{g'})).$
Suppose each $K_{g,g^{'}}$ is well defined and
a continuous function of the parameter $(a_{A},a_{B},F)$. 

For $g\in\{A,B\}$, $F\in\mathcal{F}_{g}$, let $U_{g}(a_{i},a_{-i};F)$
be the expected payoffs of the strategy profile $(a_{i},a_{-i})$
for $i$ when consequences are drawn according to $F.$ Assume $U_{A},U_{B}$
are continuous.

Suppose for every model $\Theta_{g}$ and every $(a_{A},a_{B},F)\in\Theta_{g}$
and $\epsilon>0,$ there exists an open neighborhood $V\subseteq\Theta_{g}$
of $(a_{A},a_{B},F)$, so that for every $(\hat{a}_{A},\hat{a}_{B},\hat{F})\in V$,
$1-\epsilon\le[f(a_{i},a_{A})(y)\cdot\varphi^{\bullet}(a_{A})(m)]/[\hat{f}(a_{i},\hat{a}_{A})(y)\cdot\varphi^{\bullet}(\hat{a}_{A})(m)]\le1+\epsilon$
and $1-\epsilon\le[f(a_{i},a_{B})(y)\cdot\varphi^{\bullet}(a_{B})(m)]/[\hat{f}(a_{i},\hat{a}_{B})(y)\cdot\varphi^{\bullet}(\hat{a}_{B})(m)]\le1+\epsilon$
for all $a_{i}\in\mathbb{A},y\in\mathbb{Y},m\in\mathbb{M}$. Also
suppose there is some $C>0$ so that $\ln(f(a_{i},a_{A})(y)\cdot\varphi^{\bullet}(a_{A})(m))$
and $\ln(f(a_{i},a_{B})(y)\cdot\varphi^{\bullet}(a_{B})(m))$ are
bounded in $[-C,C]$ for all $(a_{A},a_{B},F)\in\Theta_{g}$, $a_{i},a_{-i}\in\mathbb{A},y\in\mathbb{Y},m\in\mathbb{M}$.

\subsection{Learning Environment}

We first consider an environment with only one true situation, $|\mathcal{G}|=1.$
Time is discrete and infinite, $t=0,1,2,...$ A unit mass of agents,
$i\in[0,1]$, enter the society at time 0. A $p_{A}\in(0,1)$ measure
of them are assigned to model $A$ and the rest are assigned to model
$B$. Each agent born into model $g$ starts with the same full support
prior over the model, $\mu_{g}^{(0)}\in\Delta(\Theta_{g})$, and believes
there is some $(a_{A},a_{B},F)\in\Theta_{g}$ so that every group
$g$ opponent always plays $a_{g}$ and the consequences are always
generated by $F$.

In each period $t$, agents are matched up uniformly at random to
play the stage game. Each person in group $g$ has $p_{g}$ chance
of matching with someone from group $g,$ and matches with someone
from group $-g$ with the complementary chance. Each agent $i$ observes
their opponent's group membership and chooses a strategy $a_{i}^{(t)}\in\mathbb{A}$.
At the end of the match, the agent observes own consequence $y_{i}^{(t)}$
and a monitoring signal $m_{i}^{(t)}\in\mathbb{M}$ about the opponent's
play, where $m_{i}^{(t)}$ is drawn from the distribution $\varphi^{\bullet}(a_{-i})$
if their opponent uses strategy $a_{-i}$. One example of this would
be $\mathbb{M}=\mathbb{A}$ and $m_{i}^{(t)}$ is equal to the opponent's
strategy with probability $\tau\in[0,1)$ and is uniformly random
on $\mathbb{M}$ with the complementary probability. Our results for
the case when $\tau$ is close enough to 1 and each model has the
form $\Theta_{g}=\mathbb{A}^{2}\times\mathcal{F}_{g}$ provide a foundation
for EZs in environments with strategic certainty. 

The space of histories from one period is $\{A,B\}\times\mathbb{A}\times\mathbb{Y}\times\mathbb{M}$,
with typical element $(g_{i}^{(t)},a_{i}^{(t)},y_{i}^{(t)},m_{i}^{(t)})$.
It records the group membership of $i$'s opponent $g_{i}^{(t)}$,
$i$'s strategy $a_{i}^{(t)},$  $i$'s consequence $y_{i}^{(t)}$,
and $i$'s monitoring signal about the matched opponent's strategy,
$m_{i}^{(t)}$. Let $\mathbb{H}$ denote the space of all finite-length
histories.

Given the assumption on the two models, there is a well-defined Bayesian
belief operator for each model $g,$ $\mu_{g}:\mathbb{H}\to\Delta(\Theta_{g}),$
mapping every finite-length history into a belief over parameters
in $\Theta_{g}$, starting with the prior $\mu_{g}^{(0)}.$

We also take as exogenously given policy functions for choosing strategies
after each history. That is, $\mathfrak{a}_{g,g^{'}}:\mathbb{H}\to\mathbb{A}$
for every $g,g^{'}\in\{A,B\}$ gives the strategy that a group $g$
agent uses against a group $g^{'}$ opponent after every history.
Assume these policy functions are asymptotically myopic. \begin{assumption}
\label{assu:asymptotic_myopia}For every $\epsilon>0,$ there exists
$N$ so that for any history $h$ containing at least $N$ matches
against opponents of each group, $\mathfrak{a}_{g,g^{'}}(h)$ is an
$\epsilon$-best response to the Bayesian belief $\mu_{g}(h)$. \end{assumption}
From the perspective of each agent $i$ in group $g,$ $i$'s play
against groups A and B, as well as $i$'s belief over $\Theta_{g},$
is a stochastic process $(\tilde{a}_{iA}^{(t)},\tilde{a}_{iB}^{(t)},\tilde{\mu}_{i}^{(t)})_{t\ge0}$
valued in $\mathbb{A}\times\mathbb{A}\times\Delta(\Theta_{g}).$ The
randomness is over the groups of opponents matched with in different
periods, the strategies they play, and the random consequences and
monitoring signals drawn at the end of the matches. Since there is
a continuum of agents, the distribution over histories within each
population in each period is deterministic. As such, there is a deterministic
sequence $(\alpha_{AA}^{(t)},\alpha_{AB}^{(t)},\alpha_{BA}^{(t)},\alpha_{BA}^{(t)},\nu_{A}^{(t)},\nu_{B}^{(t)})\in\Delta(\mathbb{A})^{4}\times\Delta(\Delta(\Theta_{A}))\times\Delta(\Delta(\Theta_{B}))$
that describes the distributions of play and beliefs that prevail
in the two sub-populations in every period $t.$

\subsection{Steady State Limits are Equilibrium Zeitgeists}

We state and prove the learning foundation of EZs. For $(\alpha^{(t)})_{t}$
a sequence valued in $\Delta(\mathbb{A})$ and $a^{*}\in\mathbb{A},$
$\alpha^{(t)}\to a^{*}$ means $\mathbb{E}_{\hat{a}\sim\alpha^{(t)}}\parallel\hat{a}-a^{*}\parallel\to0$
as $t\to\infty$. For $(\nu^{(t)})_{t}$ a sequence valued in $\Delta(\Delta(\Theta_{g}))$
and $\mu^{*}\in\Delta(\Theta_{g}),$ $\nu^{(t)}\to\mu^{*}$ means
$\mathbb{E}_{\hat{\mu}\sim\nu^{(t)}}\parallel\hat{\mu}-\mu^{*}\parallel\to0$
as $t\to\infty.$ 
\begin{prop}
\label{prop:learning}Suppose the regularity assumptions in Appendix
\ref{subsec:Regularity-Assumptions-learning} hold, and suppose Assumption
\ref{assu:asymptotic_myopia} holds. Suppose there exists $(a_{AA}^{*},a_{AB}^{*},a_{BA}^{*},a_{BB}^{*},\mu_{A}^{*},\mu_{B}^{*})\in\mathbb{A}^{4}\times\Delta(\Theta_{A})\times\Delta(\Theta_{B})$
so that $(\alpha_{AA}^{(t)},\alpha_{AB}^{(t)},\alpha_{BA}^{(t)},\alpha_{BA}^{(t)},\nu_{A}^{(t)},\nu_{B}^{(t)})\to(a_{AA}^{*},a_{AB}^{*},a_{BA}^{*},a_{BB}^{*},\mu_{A}^{*},\mu_{B}^{*})$
and for each agent $i$ in group $g,$ almost surely $(\tilde{a}_{iA}^{(t)},\tilde{a}_{iB}^{(t)},\tilde{\mu}_{i}^{(t)})\to(a_{gA}^{*},a_{gB}^{*},\mu_{g}^{*})$.
Then, $(a_{AA}^{*},a_{AB}^{*},a_{BA}^{*},a_{BB}^{*},\mu_{A}^{*},\mu_{B}^{*})$
is an equilibrium zeitgeist.

Suppose further that for each $g,$ the model $\Theta_{g}$ has the
form $\mathbb{A}^{2}\times\mathcal{F}_{g}$. There exists some $\underline{\tau}<1$
so that for every $\tau\in(\underline{\tau},1)$ and $(a_{AA}^{*},a_{AB}^{*},a_{BA}^{*},a_{BB}^{*},\mu_{A}^{*},\mu_{B}^{*})$
satisfying the above conditions, we have that $\mu_{A}^{*}$ puts
probability 1 on $(a_{AA}^{*},a_{AB}^{*})$, $\mu_{B}^{*}$ puts probability
1 on $(a_{BA}^{*},a_{BB}^{*})$. 
\end{prop}

\begin{proof}
For $\mu$ a belief and $g\in\{A,B\},$ let $u^{\mu}(a_{i};g)$ represent
subjective expected payoff from playing $a_{i}$ against group $g$.
Suppose $a_{AA}^{*}\notin\text{argmax}_{\hat{a}\in\mathbb{A}}u^{\mu_{A}^{*}}(\hat{a};A)$
(the other cases are analogous). By the continuity assumptions on
$U_{A}$ (which is also bounded because $\mathcal{F}_{A}$ is bounded),
there are some $\epsilon_{1},\epsilon_{2}>0$ so that whenever $\mu_{i}\in\Delta(\Theta_{A})$
with $\parallel\mu_{i}-\mu_{A}^{*}\parallel<\epsilon_{1}$, we also
have $u^{\mu_{i}}(a_{AA}^{*};A)<\max_{\hat{a}\in\mathbb{A}}u^{\mu_{i}}(\hat{a};A)-\epsilon_{2}.$
By the definition of asymptotically empirical best responses, find
$N$ so that $\mathfrak{a}_{A,A}(h)$ must be a myopic $\epsilon_{2}$-best
response when there are at least $N$ periods of matches against A
and B. Agent $i$ has a strictly positive chance to match with groups
A and B in every period. So, at all except a null set of points in
the probability space, $i$'s history eventually records at least
$N$ periods of play by groups A and B. Also, by assumption, almost
surely $\tilde{\mu}_{i}^{(t)}\to\mu_{A}^{*}.$ This shows that by
asymptotically myopic best responses, almost surely $\tilde{a}_{iA}^{(k)}\not\to a_{AA}^{*},$
a contradiction.

Now suppose some $\theta_{A}^{*}=(a_{A}^{*},a_{B}^{*},f^{*})$ in
the support of $\mu_{A}^{*}$ does not minimize the weighted KL divergence
in the definition of EZ (the case of a parameter $\theta_{B}^{*}$
in the support of $\mu_{B}^{*}$ not minimizing is similar). Then
we have 

\begin{align*}
\theta_{A}^{*}\notin\underset{\hat{\theta}\in\Theta_{A}}{\text{argmin}}\left\{ \begin{array}{c}
(p_{A})\cdot D_{KL}(F^{\bullet}(a_{AA}^{*},a_{AA}^{*})\times\varphi^{\bullet}(a_{AA}^{*})\parallel\hat{F}(a_{AA}^{*},\hat{a}_{A})\times\varphi^{\bullet}(\hat{a}_{A}))\\
+(1-p_{A})\cdot D_{KL}(F^{\bullet}(a_{AB}^{*},a_{BA}^{*})\times\varphi^{\bullet}(a_{BA}^{*})\parallel\hat{F}(a_{AB}^{*},\hat{a}_{B})\times\varphi^{\bullet}(\hat{a}_{B}))
\end{array}\right\} 
\end{align*}
where $\hat{\theta}=(\hat{a}_{A},\hat{a}_{B},\hat{F}).$

This is equivalent to: 

\[
\theta_{A}^{*}\notin\underset{\hat{\theta}\in\Theta_{A}}{\text{argmin}}\left[\begin{array}{c}
(p_{A})\cdot\mathbb{E}_{(y,m)\sim F^{\bullet}(a_{AA}^{*},a_{AA}^{*})\times\varphi^{\bullet}(a_{AA}^{*})}\ln(\hat{f}(a_{AA}^{*},\hat{a}_{A})(y)\cdot\varphi^{\bullet}(\hat{a}_{A})(m))\\
+(1-p_{A})\cdot\mathbb{E}_{(y,m)\sim F^{\bullet}(a_{AB}^{*},a_{BA}^{*})\times\varphi^{\bullet}(a_{BA}^{*})}\ln(\hat{f}(a_{AB}^{*},\hat{a}_{B})(y)\cdot\varphi^{\bullet}(\hat{a}_{B})(m))
\end{array}\right]
\]

Let this objective, as a function of $\hat{\theta}$, be denoted $WL(\hat{\theta}).$
There exists $\theta_{A}^{opt}=(a_{A}^{opt},a_{B}^{opt},f^{opt})\in\Theta_{A}$
and $\delta,\epsilon>0$ so that $(1-\delta)WL(\theta_{A}^{opt})-2\delta C-3\epsilon>(1-\delta)WL(\theta_{A}^{*}).$
By assumption on the primitives, find open neighborhoods $V^{opt}$
and $V^{*}$ of $\theta_{A}^{opt},\theta_{A}^{*}$ respectively, so
that for all $a_{i}\in\mathbb{A},$ $g\in\{A,B\},$ $y\in\mathbb{Y}$,
$m\in\mathbb{M}$, $1-\epsilon\le[f^{opt}(a_{i},a_{g}^{opt})(y)\cdot\varphi^{\bullet}(a_{g}^{opt})(m)]/[\hat{f}(a_{i},\hat{a}_{g})(y)\cdot\varphi^{\bullet}(\hat{a}_{g})(m)]\le1+\epsilon$,
for all $\hat{\theta}=(\hat{a}_{A},\hat{a}_{B},\hat{f})\in V^{opt}$,
and also $1-\epsilon\le[f^{*}(a_{i},a_{g}^{*})(y)\cdot\varphi^{\bullet}(a_{g}^{*})(m)]/[\hat{f}(a_{i},\hat{a}_{g})(y)\cdot\varphi^{\bullet}(\hat{a}_{g})(m)]\le1+\epsilon$
for all $\hat{\theta}=(\hat{a}_{A},\hat{a}_{B},\hat{f})\in V^{*}$.
Also, by convergence of play in the populations, find $T_{1}$ so
that in all periods $t\ge T_{1},$ $\alpha_{AA}^{(t)}(a_{AA}^{*})\ge1-\delta$
and $\alpha_{BA}^{(t)}(a_{BA}^{*})\ge1-\delta$.

Consider a probability space defined by $\Omega:=(\{A,B\}\times\mathbb{A}^{2}\times(\mathbb{Y})^{\mathbb{A}^{2}}\times\mathbb{M}^{\mathbb{A}})^{\infty}$
that describes the randomness in an agent's learning process. For
a point $\omega\in\Omega$ and each period $t\ge1$, $\omega_{t}=(g,a_{-i,A},a_{-i,B},(y_{a_{i},a_{-i}})_{(a_{i},a_{-i})\in\mathbb{A}^{2}},(m_{a_{-i}})_{a_{-i}\in\mathbb{A}})$
specifies the group $g$ of the matched opponent, the play $a_{-i,A},a_{-i,B}$
of hypothetical opponents from groups A and B, the hypothetical consequence
$y_{a_{i},a_{-i}}$ that would be generated for every pair of strategies
$(a_{i},a_{-i})$ played, and the hypothetical monitoring signal $m_{a_{-i}}$
that would be generated for every opponent strategy $a_{-i}$. As
notation, let $o(\omega,t)$, $a_{-i,A}(\omega,t),$ $a_{-i,B}(\omega,t)$,
$y_{a_{i},a_{-i}}(\omega,t)$, $m_{a_{-i}}(\omega,t)$ denote the
corresponding components of $\omega_{t}.$ For $T_{2}\ge T_{1},$
define $\mathbb{P}_{T_{2}}$ over $\Omega_{T_{2}}^{\infty}:=\times_{t=T_{2}}^{\infty}(\{A,B\}\times\mathbb{A}^{2}\times(\mathbb{Y})^{\mathbb{A}^{2}}\times\mathbb{M}^{\mathbb{A}})$
in the natural way. That is, it is independent across periods, and
within each period, the density (or probability mass function if $\mathbb{Y}$
is finite) of $\omega_{t}=(g,a_{-i,A},a_{-i,B},(y_{a_{i},a_{-i}})_{(a_{i},a_{-i})\in\mathbb{A}^{2}},(m_{a_{-i}})_{a_{-i}\in\mathbb{A}})$
is 
\[
p_{g}\cdot\alpha_{AA}^{(t)}(a_{-i,A})\alpha_{BA}^{(t)}(a_{-i,B})\cdot\prod_{(a_{i},a_{-i})\in\mathbb{A}^{2}}f^{\bullet}(a_{i},a_{-i})(y_{a_{i},a_{-i}})\cdot\prod_{a_{-i}\in\mathbb{A}}\varphi^{\bullet}(a_{-i})(m_{a_{-i}}).
\]

For $\theta=(a_{A}^{\theta},a_{B}^{\theta},F^{\theta})\in\Theta_{A}$
with $f^{\theta}$ the density of $F^{\theta}$, $\omega\in\Omega_{T_{2}}^{\infty},$
consider the process in $s=1,2,3,...$
\begin{multline*} 
\ell_{s}(\theta,\omega):=\frac{1}{s}\sum_{t=T_{2}+1}^{T_{2}+s}\ln[f^{\theta}(a_{A,o(\omega,t)}^{*},a_{o(\omega,t)}^{\theta})(y_{a_{A,o(\omega,t)}^{*},a_{-i,o(\omega,t)}(\omega,t)}(\omega,t))\\ \cdot\varphi^{\bullet}(a_{-i,o(\omega,t)}(\omega,t))(m_{a_{-i,o(\omega,t)}(\omega,t)}(\omega,t))].
\end{multline*}
By choice of the neighborhood $V^{*},$ for every $s$, 
\begin{align*}
\sup_{\theta_{A}\in V^{*}}\ell_{s}(\theta_{A},\omega) & \le\epsilon+\frac{1}{s}\sum_{t=T_{2}+1}^{T_{2}+s}\ln[f^{\theta}(a_{A,o(\omega,t)}^{*},a_{o(\omega,t)}^{*})(y_{a_{A,o(\omega,t)}^{*},a_{-i,o(\omega,t)}(\omega,t)}(\omega,t)) \\ & \hspace{50mm} \cdot\varphi^{\bullet}(a_{-i,o(\omega,t)}(\omega,t))(m_{a_{-i,o(\omega,t)}(\omega,t)}(\omega,t))]\\
   &   \le\epsilon+\frac{1}{s}\sum_{t=T_{2}+1}^{T_{2}+s}
1_{\{a_{-i,o(\omega,t)}(\omega,t)=a_{o(\omega,t),A}^{*}\}}\cdot\ln[f^{\theta}(a_{A,o(\omega,t)}^{*},a_{o(\omega,t)}^{*})(y_{a_{A,o(\omega,t)}^{*},a_{o(\omega,t),A}^{*}}(\omega,t)) \\ & \hspace{50mm} \cdot \varphi^{\bullet}(a_{o(\omega,t),A}^{*})(m_{a_{o(\omega,t),A}^{*}}(\omega,t))]\\ & \hspace{26mm}
+ (1-1_{\{a_{-i,o(\omega,t)}(\omega,t)=a_{o(\omega,t),A}^{*}\}})\cdot C.
\end{align*} 
Since $T_{2}\ge T_{1},$ in every period $t,$ $\mathbb{P}_{T_{2}}(a_{-i,o(\omega,t)}(\omega,t)=a_{o(\omega,t),A}^{*})\ge1-\delta$.
Let $(\xi_{k})_{k\ge1}$ a related stochastic process: it is i.i.d.
such that each $\xi_{k}$ has $\delta$ chance to be equal to $C,$
$(1-\delta)p_{A}$ chance to be distributed according to $\ln(f^{*}(a_{AA}^{*},a_{A}^{*})(y)\cdot\varphi^{\bullet}(a_{A}^{*})(m))$
where $y\sim f^{\bullet}(a_{AA}^{*},a_{AA}^{*})$ and $m\sim\varphi^{\bullet}(a_{AA}^{*}),$
and $(1-\delta)p_{B}$ chance to be distributed according to $\ln(f^{*}(a_{AB}^{*},a_{B}^{*})(y)\cdot\varphi^{\bullet}(a_{B}^{*})(m))$
where $y\sim f^{\bullet}(a_{AB}^{*},a_{BA}^{*})$ and $m\sim\varphi^{\bullet}(a_{BA}^{*}).$
By law of large numbers, $\frac{1}{s}\sum_{k=1}^{s}\xi_{k}$ converges
almost surely to $\delta C+(1-\delta)WL(\theta_{A}^{*}).$ By this
comparison, $\limsup_{s}\sup_{\theta_{A}\in V^{*}}\ell_{s}(\theta_{A},\omega)\le\epsilon+\delta C+(1-\delta)WL(\theta_{A}^{*})$
$\mathbb{P}_{T_{2}}$-almost surely. By a similar argument, $\liminf_{s}\inf_{\theta_{A}\in V^{opt}}\ell_{s}(\theta_{A},\omega)\ge-\epsilon-\delta C+(1-\delta)WL(\theta_{A}^{opt})$
$\mathbb{P}_{T_{2}}$-almost surely.

Along any $\omega$ where we have both $\limsup_{s}\sup_{\theta_{A}\in V^{*}}\ell_{s}(\theta_{A},\omega)\le\epsilon+\delta C+(1-\delta)WL(\theta_{A}^{*})$
and $\liminf_{s}\inf_{\theta_{A}\in V^{opt}}\ell_{s}(\theta_{A},\omega)\ge-\epsilon-\delta C+(1-\delta)WL(\theta_{A}^{opt})$,
if $\omega$ also leads to $i$ always playing $a_{AA}^{*}$ against
group A and $a_{AB}^{*}$ against group B in all periods starting
with $T_{2}+1,$ then the posterior belief assigns to $V^{*}$ must
tend to 0, hence $\tilde{\mu}_{i}^{(t)}\not\to\mu_{A}^{*}.$ Starting
from any length $T_{2}$ history $h,$ there exists a subset $\hat{\Omega}_{h}\subseteq\Omega_{T_{2}}^{\infty}$
that leads to $i$ not playing the EZ strategy in at least one period
starting with $T_{2}+1.$ So conditional on $h,$ the probability
of $\tilde{\mu}_{i}^{(t)}\to\mu_{A}^{*}$ is no larger than $1-\mathbb{P}_{T_{2}}(\hat{\Omega}_{h}).$
The unconditional probability is therefore no larger than $\mathbb{E}_{h}[1-\mathbb{P}_{T_{2}}(\hat{\Omega}_{h})],$
where $\mathbb{E}_{h}$ is taken with respect to the distribution
of period $T_{2}$ histories for $i.$ But this term is also the probability
of $i$ playing non-EZ action at least once starting with period $T_{2}.$
Since there are finitely many actions and $(\tilde{a}_{iA}^{(t)},\tilde{a}_{iB}^{(t)})\to(a_{AA}^{*},a_{AB}^{*})$
almost surely, $\mathbb{E}_{h}[1-\mathbb{P}_{T_{2}}(\hat{\Omega}_{h})]$
tends to 0 as $T_{2}\to\infty.$ We have a contradiction as this shows
$\tilde{\mu}_{i}^{(t)}\not\to\mu_{A}^{*}$ with probability 1.

Now we prove the second part of this proposition. Let $\bar{K}<\infty$
be an upper bound on $D_{KL}(F^{\bullet}(a_{i},a_{-i})\parallel\hat{F}(a_{i},\hat{a}_{-i}))$
across all $a_{i},a_{-i}\in\mathbb{A},$ $(\hat{a}_{A},\hat{a}_{B},\hat{F})\in\Theta_{A}\cup\Theta_{B}$.
Here $\bar{K}$ is finite because $\mathbb{A}$ is finite and $K_{g,g^{'}}$
is continuous in the parameter, which is from a compact domain. Let
$\varphi_{\tau}^{\bullet}(a_{-i})$ be the distribution over $\mathbb{M}$
when opponent plays $a_{-i}$ and the monitoring structure is such
that the monitoring signal matches the opponent's strategy with probability
$\tau$ and is uniformly random on $\mathbb{M}$ with the complementary
probability. It is clear that there exists some $\underline{\tau}<1$
so that for any $a_{-i}\ne a_{-i}^{'}$, $\tau\in(\underline{\tau},1),$
we get $\min(p_{A},p_{B})\cdot D_{KL}(\varphi_{\tau}^{\bullet}(a_{-i})\parallel\varphi_{\tau}^{\bullet}(a_{-i}'))>\bar{K}.$
Therefore, given any $(a_{AA}^{*},a_{AB}^{*},a_{BA}^{*})\in\mathbb{A}^{3},$
the solution to 
\[
\underset{\hat{\theta}\in\Theta_{A}}{\min}\left[\begin{array}{c}
(p_{A})\cdot D_{KL}(F^{\bullet}(a_{AA}^{*},a_{AA}^{*})\times\varphi^{\bullet}(a_{AA}^{*})\parallel\hat{F}(a_{AA}^{*},\hat{a}_{A})\times\varphi^{\bullet}(\hat{a}_{A}))\\
+(1-p_{A})\cdot D_{KL}(F^{\bullet}(a_{AB}^{*},a_{BA}^{*})\times\varphi^{\bullet}(a_{BA}^{*})\parallel\hat{F}(a_{AB}^{*},\hat{a}_{B})\times\varphi^{\bullet}(\hat{a}_{B}))
\end{array}\right]
\]
must satisfy $\hat{a}_{A}=a_{AA}^{*},$ $\hat{a}_{B}=a_{BA}^{*}$,
because $(a_{AA}^{*},a_{BA}^{*},F)$ for any $F\in\Theta_{A}$ has
a KL divergence no larger than $\bar{K}$. On the other hand, any
$(\hat{a}_{A},\hat{a}_{B},\hat{F})$ with either $\hat{a}_{A}\ne a_{AA}^{*}$
or $\hat{a}_{B}\ne a_{BA}^{*}$ has KL divergence strictly larger
than $\bar{K}$ by the choice of $\tau$. 
\end{proof}

\subsection{Multiple Situations}

Now suppose there are multiple situations $G\in\mathcal{G}$ and a
distribution $q\in\Delta(\mathcal{G})$, with $\mathcal{G}$ finite.
At the start of period $t=1,$ Nature draws a situation $G^{(1)}$
from $\mathcal{G}$ according to $q$, and consequences are generated
according to $F^{\bullet}(\cdot,\cdot,G^{(1)})$ until period $t=T+1.$
In period $T+1,$ Nature again draws a situation $G^{(2)}$ from $\mathcal{G}$
according to $q$, and consequences are generated according to $F^{\bullet}(\cdot,\cdot,G^{(2)})$
until period $t=2T+1,$ and so forth. Agents start with a prior over parameters in
their group's  model, $\mu_{g}^{(0)}\in\Delta(\Theta_{g})$.
In periods $T+1,2T+1,...$ agents reset their belief to $\mu_{g}^{(0)},$
and their belief in each period over the  parameters in their
 model only use histories since the last reset. This belief
corresponds to agents thinking that the data-generating process is
redrawn according to $\mu_{g}^{(0)}$ every $T$ periods.

Suppose for every $G\in\mathcal{G},$ the hypotheses of Proposition
\ref{prop:learning} hold in a society where $G$ is the only true
situation. Denote $(a_{AA}^{*}(G),a_{AB}^{*}(G),a_{BA}^{*}(G),a_{BB}^{*}(G),\mu_{A}^{*}(G),\mu_{B}^{*}(G))$
as the limit of the agents' behavior and beliefs with situation $G.$
Then it is straightforward to see that in a society with the situation
redrawn every $T$ periods, the expected undiscounted average payoff
of an agent in group $g$ approaches the fitness of $g$ in the EZ
characterized by the behavior and beliefs $(a_{AA}^{*}(G),a_{AB}^{*}(G),a_{BA}^{*}(G),a_{BB}^{*}(G),\mu_{A}^{*}(G),\mu_{B}^{*}(G))_{G\in\mathcal{G}}$
with the distribution $q$ over situations, as $T\to\infty$. This
provides a foundation for fitness in EZ as the agents' objective payoffs
when the true situation changes sufficiently slowly. 
\end{document}